\documentclass[a4paper,UKenglish,cleveref,autoref,thm-restate]{lipics-v2021}

\hideLIPIcs  

\usepackage{amsmath}
\usepackage{amsfonts}
\usepackage{amsthm}
\usepackage{tikz}
\usepackage{todonotes}
\usetikzlibrary{automata, positioning, arrows}
\usetikzlibrary{patterns}

\usepackage{blindtext}

\newtheorem{openproblem}{Open Problem}

\DeclareMathOperator{\RE}{RE}
\DeclareMathOperator{\REG}{REG}
\DeclareMathOperator{\VPL}{VPL}
\DeclareMathOperator{\CF}{DCF}

\DeclareMathOperator{\REGlang}{REGLang}
\DeclareMathOperator{\VPLlang}{VPLang}
\DeclareMathOperator{\CFlang}{DCFLang}
\DeclareMathOperator{\GCFlang}{CFLang}
\DeclareMathOperator{\WE}{WE}

\DeclareMathOperator{\LEN}{LEN}
\DeclareMathOperator{\CON}{CON}
\DeclareMathOperator{\+}{+}
\DeclareMathOperator{\ra}{\rightarrow}
\DeclareMathOperator{\nat}{\mathbb{N}}

\title{Formal Languages via Theories over Strings}

\author{Joel {Day}}{Loughborough University, United Kingdom}{J.Day@lboro.ac.uk}{}{}

\author{Vijay {Ganesh}}{University of Waterloo, Canada}{vijay.ganesh@uwaterloo.ca}{}{}

\author{Nathan {Grewal}}{University of Waterloo, Canada}{negrewal@uwaterloo.ca}{}{}

\author{Florin {Manea}}{Universität Göttingen, Germany}{florin.manea@cs.informatik.uni-goettingen.de}{}{}

\authorrunning{J. Day, V. Ganesh, N. Grewal, and F. Manea} 

\Copyright{Joel Day, Vijay Ganesh, Nathan Grewal, and Florin Manea} 



\EventEditors{John Q. Open and Joan R. Access}
\EventNoEds{2}
\EventLongTitle{42nd Conference on Very Important Topics (CVIT 2016)}
\EventShortTitle{CVIT 2016}
\EventAcronym{CVIT}
\EventYear{2016}
\EventDate{December 24--27, 2016}
\EventLocation{Little Whinging, United Kingdom}
\EventLogo{}
\SeriesVolume{42}
\ArticleNo{23}

\begin{document}
\nolinenumbers

\maketitle

\begin{abstract}
We investigate the properties of formal languages expressible in terms of formulas over quantifier-free theories of word equations, arithmetic over length constraints, and language membership predicates for the classes of regular, visibly pushdown, and deterministic context-free languages. In total, we consider 20 distinct theories and decidability questions for problems such as emptiness and universality for formal languages over them. First, we discuss their relative expressive power and observe a rough division into two hierarchies based on whether or not word equations are present. Second, we consider the decidability status of several important decision problems, such as emptiness and universality. Note that the emptiness problem is equivalent to the satisfiability problem over the corresponding theory. Third, we consider the problem of whether a language in one theory is expressible in another and show several negative results in which this problem is undecidable. These results are particularly relevant in the context of normal forms in both practical and theoretical aspects of string solving. \looseness=-1


\end{abstract}

\section{Introduction}
Logical theories based on strings (or words) over a finite alphabet have been an important topic of study for decades. Connections to arithmetic (see e.g. Quine~\cite{quine1946}) and fundamental questions about free (semi)groups underpinned interest in logics involving concatenation and equality. Combining these two things leads to word equations: expressions $\alpha = \beta$ where $\alpha$ and $\beta$ are terms obtained by concatenating variables and concrete words over a finite alphabet. For example, if $x$ and $y$ are variables, and our alphabet is $\Sigma = \{a,b\}$, then $x ab y = y ba x$ is a word equation. Its solutions are variable-substitutions unifying the two sides: $ x \to bb, y \to b$ would be one such solution in the previous example. \looseness=-1

The existential theory of a finitely generated free semigroup consists of formulas made up of Boolean combinations of word equations. In fact, the problem of deciding whether such a formula is true is equivalent to determining satisfiability of word equations, since any such formula can be transformed into a single word equation without disrupting satisfiability (see~\cite{lot:alg,karhumaki2000}). 
Famously, 
Makanin showed in 1977 that satisfiability of word equations is algorithmically decidable~\cite{makanin1977}.
Since then, several improvements to the algorithm proposed by Makanin have been discovered:  Plandowski~\cite{plandowski1999} showed that the problem could be solved in PSPACE, which has later been refined to nondeterministic linear space by Je{\.z}~\cite{Jez2017}. Schulz~\cite{schulz1990} showed that the problem remains decidable even when the variables are constrained by  regular languages, limiting the possible substitutions (see also Chapter 12 of~\cite{lot:alg}). On the other hand, if length constraints (requiring that some pairs of variables are substituted for words of the same length) are permitted, then the (un)decidability of the problem is a long-standing open problem. \looseness=-1

Word equations, and logics involving strings more generally, have remained a topic of interest within the Theoretical Computer Science community, in particular due to their fundamental role within Combinatorics on Words and Formal Languages, and more recently due to interested from the Formal Methods community. The latter can be attributed to increasing popularity and influence of software tools called string-solvers, which seek to algorithmically solve constraint problems involving strings~\cite{amadini2021survey,hague2019strings}. In this setting, a string constraint is a property or piece of information about an unknown string, and the string solvers try to determine whether strings exist which satisfy combinations of string constraints of various types. Word equations, regular language membership, and comparisons between lengths are all among the most prominent building blocks of string constraints, and when combined are sufficient to model several others. String-solvers are also useful in other areas like Database Theory, particularly e.g. for evaluating path queries in graph databases~\cite{barcelo2017graph} and in connection with Document Spanners~\cite{Frey:doc,Frey:doc2,frey:finite}.

A wealth of string-solvers are now available~\cite{ss1,berzish2021smt,z3str4,kiezun2009,abdulla2015,trinh2016progressive,barrett2011,abdulla2020efficient} (see~\cite{amadini2021survey,hague2019strings} for an overview). However, the underlying task of determining the satisfiability of string constraints remains a challenging problem and a barrier to more effective implementations. 
Some results already exist addressing the computability/ complexity and expressibility of combinations of string constraints. \cite{majumdar2021quadratic,Ganesh2012WordEW,le2018decision,lin2016string,liang2015decision} identify restrictions on word equations which result in a decidable satisfiability problem even when length constraints are present. Several further ways of augmenting word equations (i.e., additional predicates or constraints on the variables), are shown to be undecidable  in~\cite{chen2017decidable,day:rp18,buchi1990definability,halfon2017decidability}. 
Moreover, Büchi and Senger~\cite{buchi1990definability} 
considered definability in the theory of concatenation and showed, on the one hand, that length is not definable using equality and concatenation alone, and, on the other hand, that if predicates are present, which count occurrences of at least two different letters,
we obtain an undecidable theory. Karhum\"aki, Plandowski and Mignosi~\cite{karhumaki2000} considered explicitly the question of which formal languages are expressible as the set of solutions to a word equation, projected onto a single variable. They show that several simple languages are not expressible, including some simple regular languages. However, they do not consider the case when word equations may be augmented with additional constraints, which can drastically change the languages expressed.\looseness=-1

Nevertheless, despite results such as those mentioned above, little is known about the true expressive power of word equations and of string logics involving word equations in conjunction with other common types of string constraints. A greater understanding in this regard would help in settling open problems (such as for whether satisfiability for word equations with length constraints is decidable), and also with devising string solving strategies: often simply finding a solution to one constraint is not enough and the set of solutions must be considered more generally in order to account for other constraints which might be present, or to determine that no solution exists. Moreover, a common tactic is to rewrite constraints into some normal form before solving and understanding when and how this can be done also requires knowledge of the relative expressive power of subsets of constraints. \looseness=-1

This work aims to fill gaps in the understanding of the properties and expressivity of some of the most important combinations of string constraints by considering languages expressible in the sense of~\cite{karhumaki2000}. In this regard, our results can be seen as extending~\cite{karhumaki2000} to a more general (and more practical) setting. As such, our wider context requires a range of new approaches and leads to new questions and insights not considered in previous works.

{\bf Our framework:}
We consider a landscape of string-based logics incorporating various types of atoms inspired by and strongly related to prominent varieties of string-constraints. In particular, we consider logics with different combinations of: equality between strings,  concatenations of strings, membership of formal languages, and linear arithmetic over string-lengths.\looseness=-1

In total, we consider 20 distinct families of logical theories (each family containing a different theory for each possible underlying alphabet $\Sigma$), which are introduced in detail in Section~\ref{sec:theories}. Taking inspiration from~\cite{karhumaki2000} (and partly from \cite{BenediktLSS03}, where relation-definability by logics over strings was studied in a database-theory centered framework), 
we study these logics from a formal languages perspective by looking at the set of values a variable may take while preserving satisfiability of a formula. Specifically, given a formula $f$ from a quantifier-free logical theory $\mathfrak{T}$, we say that the language expressed by a variable $x$ occurring in $f$ is the set of concrete values $w$ such that substituting $x$ for $w$ in $f$ yields a satisfiable formula. In the general case, we can think of the property that the formula $f$ defines via the variable $x$. However, since we consider logics in which $x$ is substituted for finite strings, we get a formal language.

We are interested both in the expressive power of the logical theories w.r.t. what languages they can express, and in their computational properties with respect to canonical decision problems within formal languages such as emptiness, universality, equivalence and inclusion.\looseness=-1

Together, the 4 types of atoms we allow cover many of the most prominent types of string constraints, as listed in \cite{amadini2021survey}. 
While predicates related to equality between strings, concatenations of strings, and linear arithmetic over string-lengths do not need more explanations, a discussion is in order w.r.t. our choice of language membership predicates. In this case, they are considered for the classes of regular, deterministic context-free, or as an intermediary between the two, visibly pushdown languages.
While there are many classes of languages we might choose to consider between regular and deterministic context-free, there are several advantages to choosing the visibly pushdown languages in particular. Firstly, they exhibit an attractive balance of being computationally reasonable (they have many of the desirable closure and algorithmic properties of the regular languages) while simultaneously being powerful enough to provide a reasonable model in many verification and software analysis applications, in line with our motivations from string-solving. Moreover, since they directly generalise the regular languages, but with sufficient memory capabilities to model certain types of length comparisons, the combination of word equations and visibly pushdown language constraints generalise the combination of word equations with both length and regular constraints. The latter is of particular interest in the context of string-solving, but is a case for which the decidability of satisfiability remains open and is likely to be difficult to resolve. We show that satisfiability for the former is undecidable and thus that already a very limited extension to regular and length constraints is enough to reach this negative result.\looseness=-1

{\bf Our results:} Firstly, in Section~\ref{sec:separation}, we compare the relative expressive power of the different theories. On the one hand, we manage to group certain families together, where they express the same class of languages. In the technical Lemma \ref{lemma:cftheories}, we show that adding linear arithmetic over string-lengths to a theory allowing only language membership predicates for a class of languages with good language theoretic properties does not alter its expressive power. Thus, the theories in which only regular language (or visibly pushdown language) membership predicates are allowed and the theories in which length comparison is added to those membership predicates are equivalent. While in the case of theories based on regular language membership predicates we can also add concatenation without changing the expressive power, we show in Theorem~\ref{thm:VPLCONequalsRE} that adding this operation to theories based on visibly pushdown language membership predicates strictly increases their expressive power, as they can express all recursively enumerable languages. Moreover, we also provide several separation results between the classes of language expressed by various theories. One of the ways we achieve this is by non-trivially extending pumping-lemma style tools for word equations from~\cite{karhumaki2000} to our more general settings. The overall hierarchy of classes of languages expressible in our theories is depicted in Figure~\ref{fig:overviewoftheories}. \looseness=-1

While our results from Section~\ref{sec:separation} are already interesting from a language-theoretic point of view, they are also relevant for the emptiness problem for classes of languages expressed by our theories, which is equivalent to the satisfiability problem for formulas over those theories. As such, our results allow us to non-trivially extend the state-of-the-art related to the satisfiability of string constraints. In particular, we settle the previously mentioned interesting case in which word equations (in fact, even concatenation without explicit string-equality is sufficient) are combined with visibly pushdown language membership constraints. When combined with existing results, our results establish a relatively complete description of when the emptiness problem is decidable/undecidable (see the left part of Figure~\ref{fig:overviewoftheories}). The cases left open are the combinations of word equations with length constraints with or without regular constraints, which are long-standing open problems. \looseness=-1

Further, in Section~\ref{sec:universality}, we consider the universality problem and a related variant, namely the subset universality problem in which we want to test whether a language is exactly $S^*$ for a subset $S$ of the underlying alphabet. Again, our results fill in gaps in the knowledge and allow us to paint a comprehensive picture of the decidability status of these problems for our theories (see the right part of Figure~\ref{fig:overviewoftheories}). Since the universal language is expressible in all our theories, in combination with results from Section~\ref{sec:separation} and from the literature, we obtain a complete picture for the equivalence and inclusion problems. However, a substantial further benefit (and a large part of our motivation for studying this problem) is that it allows us to use Greibach's theorem in many cases (as stated in Theorem~\ref{the:greibachapplies}) to establish further undecidability results (e.g., Theorems \ref{the:eliminatinglength} and \ref{the:greibachpumping}). In particular, Theorem \ref{the:eliminatinglength} is part of a larger line of thought, developed in Section~\ref{sec:expressivity}, in which we consider the question of when it is (un)decidable if a language expressed in one theory can be expressed in another. Such problems are particularly interesting in the context of practical string solving,  because they essentially ask whether a property defined by one kind of string constraint can be algorithmically converted to another. Often, it is the combinations of different kinds of string constraints which lead to high complexities in solving, so being able to rewrite constraints in different forms can be a powerful pre-processing technique. We also identify some interesting cases where Greibach's theorem is not applicable, and thus where other approaches are needed (e.g. Theorem \ref{the:undecifWEisREG}). \looseness=-1


{\bf Summary:}
Our aim in this contribution is to obtain a more complete understanding of the computational properties and expressivity of languages expressed by various combinations of commonly occurring types of string constraints. Naturally, we are able to account for several cases by recalling, or extending existing results from literature, so at the beginning of each section, we give a single theorem that summarizes existing results and discuss their consequences. This allows us to subsequently focus on the most interesting remaining cases, many of which we are able to resolve by drawing on a range of techniques rooted in formal languages, automata theory, combinatorics on words and computability theory. Our contributed results are a substantial improvement the state of understanding of the theories considered, particularly with respect to their expressive power. In those cases we are unable to resolve, we identify several interesting new open problems and a need for novel techniques for tackling them. 


\section{Preliminaries}\label{sec:preliminaries}

Let $\mathbb{N} = \{1,2,3,\ldots\} $ and $\mathbb{N}_0 = \{0\} \cup \mathbb{N}$. Let $\mathbb{Z}$ denote the set of integers. Let $\Sigma = \{a_1,a_2,\ldots, a_n\}$ be an alphabet. We denote by $\Sigma^*$ the set of all words over $\Sigma$ including the empty word, which we denote $\varepsilon$. In other words, $\Sigma^*$ is the free monoid generated by $\Sigma$ under the operation of concatenation. For words $u,v \in \Sigma^*$ we denote their concatenation either by $u \cdot v$ or simply as $uv$. Given a set of variables $\mathcal{X} = \{x_1,x_2,\ldots\}$ and an alphabet $\Sigma$, a {\em word equation} is a pair $(\alpha,\beta) \in (X\cup \Sigma)^* \times (X\cup \Sigma)^*$, usually written as $\alpha = \beta$. A solution to a word equation is a substitution of the variables for words in $\Sigma^*$ such that both sides of the equation become identical. Formally, we model solutions as morphisms. That is, we say a substitution is a (homo)morphism $h : (X \cup \Sigma)^* \to \Sigma^*$ satisfying $h(a) = a$ for all $a \in \Sigma$, and a solution to a word equation $\alpha = \beta$ is a substitution $h$ such that $h(\alpha) = h(\beta)$.

We refer to~\cite{HopcroftUllman} for standard definitions and well known results from formal language theory regarding e.g. recursively enumerable languages ($\RE$), regular languages ($\REGlang$), context free languages ($\GCFlang$), deterministic context-free languages ($\CFlang$), finite and pushdown automata, etc. 

In addition, we refer to \cite{Alur2004,AlurKMV05,AlurM09} for background on visibly pushdown automata and visibly pushdown languages ($\VPLlang$) but also give here the main definitions. More precisely, a pushdown alphabet $\widetilde{\Sigma}$ is a triple $(\Sigma_c, \Sigma_i, \Sigma_r)$ of pairwise-disjoint alphabets known as the call, internal and return alphabets respectively. A visibly pushdown automaton (VPA) is a pushdown automaton for which the stack operations (i.e. whether a push, pop or neither is performed) are determined by the input symbol which is read. In particular, any transition for which the input symbol $a$ belongs to the call alphabet $\Sigma_c$, must push a symbol to the stack while any transition for which $a \in \Sigma_r$ must pop a symbol from the stack unless the stack is empty and any transition for which $a \in \Sigma_i$ must leave the stack unchanged. Acceptance of a word is determined by the state the automaton is in after reading the whole word. The stack does not need to be empty for a word to be accepted. A $\widetilde{\Sigma}$-visibly pushdown language is the set of words accepted by a visibly pushdown automaton with pushdown alphabet $\widetilde{\Sigma}$. A language $L$ is a visibly pushdown language (and is part of the class $\VPLlang$) if there exists a pushdown alphabet $\widetilde{\Sigma}$ such that $L$ is a $\widetilde{\Sigma}$-visibly pushdown language. The class $\VPLlang$ is a strict superset of the class of regular languages and a strict subset of the class of deterministic context-free languages, which retains many of the nice decidability and closure properties of regular languages. In particular, it is shown in~\cite{Alur2004} that $\VPLlang$ is closed under union, intersection and complement and moreover that the emptiness, universality, inclusion and equivalence probelms are all decidable for $\VPLlang$.

By a {\em theory}, we mean a set $\mathfrak{T} = \{f_1,f_2,\ldots \}$ of formulas adhering to given syntax and to which we associate a particular semantics. The theories we consider (introduced in Section \ref{sec:theories}) consist of quantifier-free formulas. The typical computational questions one might consider with respect to a given theory $\mathfrak{T}$ are {\em Satisfiability}: given formula $f \in \mathfrak{T}$, does there exist an assignment of the variables in $f$ such that $f$ becomes true under the associated semantics? and {\em Validity}: given formula $f \in \mathfrak{T}$, is $f$ true under all assignments of the variables occurring in $f$? \looseness=-1

The questions we address have a slightly different flavour: given formula $f \in \mathfrak{T}$ and variable $x$ occurring in $f$, we are interested in properties of the set of all values $w$ for which there is an assignment mapping $x$ to $w$ which makes the formula true. Thus, we consider the set of concrete values $w$ for which $f$ remains satisfiable once the variable $x$ has been replaced by $w$. Since we shall focus on theories in which variables represent words, we refer to the set of all such values $w$ as the {\em language} expressed by the variable $x$ in the formula $f$. In this respect, we extend the notion of languages expressible by word equations~\cite{karhumaki2000} to arbitrary string-based logical theories. We say a language $L$ is {\em expressed } by a formula if it contains a variable $x$ such that $L$ is the language expressed by $x$ in $f$. We say that $L$ is {\em expressible} in a theory $\mathfrak{T}$ if there exists a formula $f \in \mathfrak{T}$ and variable $x$ occurring in $f$ such that $L$ is expressed by $x$ in $f$.\looseness=-1

We shall consider typical decision problems such as emptiness and universality for languages expressed by formulas in a given theory $\mathfrak{T}$. In this context, the input is a formula $f \in \mathfrak{T}$ and a variable $x$ occurring in $f$. So, e.g, in the case of emptiness, we might be given a formula $x = aba \land x \cdot y = ababba$ along with the variable $y$, and we must decide whether the language $L_y$ expressed by $y$ in that formula is the empty set or not. In this case, $L_y = \{bba\} \not= \emptyset$ so the answer is no. Clearly, for any formula $f$ and variable $x$, the emptiness problem for the language expressed by $x$ in $f$ is equivalent to the satisfiability problem for $f$. Thus, we consider a set of problems which directly generalise the satisfiability problem.

For theories containing word equations, we shall use notions and results from~\cite{karhumaki2000} to reason about (in)expressibility of languages, such as the notion of a {\em synchronising} $\mathfrak{F}$-factorisation.\looseness=-1

\begin{definition}[\cite{karhumaki2000}]
Let $\mathfrak{F}$ be a property of words. An $\mathfrak{F}$-factorisation of a word $w$ is a factorisation $w = w_1\ldots w_n$ such that each $w_i$ has the property $\mathfrak{F}$. Moreover $\mathfrak{F}$ is {\em synchronising} if every word has exactly one $\mathfrak{F}$-factorisation and additionally there exist $l,r \in \mathbb{N}$ such that for any words $x,y$ with $\mathfrak{F}$-factorisations $x = x_1x_2\ldots x_s$ and $y = y_1 y_2 \ldots y_k$ where $k > l+r$ and $y$ is a factor of $x$, the following hold: 
\begin{enumerate}
\item there exist $l' \leq l, r' \leq r$ and $p,q$ with $1 \leq p \leq q \leq s$ such that $q-1-p = k-r'-l'$ and $x_p = y_{l'+1}$, $x_{p+1} = y_{l'+1}$... $x_q = y_{k-r'+1}$; 
\item $y_1 y_2 \ldots y_{l'}$ is a suffix of $x_{\max(1,p-l)} \ldots x_{p-1}$; 
\item $y_{k-r'+1} \ldots y_{k} $ is a prefix of $x_q \ldots x_{\min(s,q+r-1)}$.
\end{enumerate}
Intuitively, $\mathfrak{F}$ is synchronising for some parameters $l,r$ if the $\mathfrak{F}$-factorisations of a word $x$ and a factor $y$ of $x$ synchronise (or coincide) except for the first $l$ parts and the last $r$ parts.
\end{definition}

\begin{remark}
\cite{karhumaki2000} provides several examples of synchronising factorisations, including splitting a word into blocks of a single letter which is clearly synchronising.
\end{remark}

\section{Logical Theories Over Strings Constraints}\label{sec:theories}

In this section, we introduce a variety of logical theories encompassing the most common kinds of string constraints (as overviewed in \cite{amadini2021survey}). We define three sets of terms as follows. Let $\mathcal{X} = \{x_1,x_2,\ldots \}$ be an infinite set of string variables. Let $\Sigma$ be a finite alphabet. Let $\mathcal{T}_{str}^\Sigma = \mathcal{X} \cup \Sigma^*$ be the set of {\em basic string terms}. Let $\mathcal{T}_{str,con}^{\Sigma} = (\mathcal{X} \cup \Sigma)^*$ be the set of {\em extended string terms}. Note that $\mathcal{T}_{str,con}^{\Sigma}$ is the closure of $\mathcal{T}_{str}^\Sigma$ under the concatenation $(\cdot)$ operation. Let $\mathcal{T}_{arith}^\Sigma = \{k_0 + k_1|s_1| +k_2|s_2| + \ldots + k_n|s_n| \mid n \in \mathbb{N}_0, k_i \in \mathbb{Z}, \text{ and } s_i \in \mathcal{T}_{str}^\Sigma\}$ be the set of {\em length terms}. We interpret $|s|$ as the length of the string term $s$, so $\mathcal{T}_{arith}^\Sigma$ is the set of linear combinations of lengths of string terms. Note that since we can express the length of a concatenation of string terms as a linear combination of lengths of basic string terms, it is no restriction that $s_i \in \mathcal{T}_{str}^\Sigma$ rather than $\mathcal{T}_{str,con}^\Sigma$ (this allows us to consider theories containing length terms both with and without concatenation). We construct three types of atoms from terms as follows:\looseness=-1

\noindent (A1) Language membership constraints of the form $s \in L$ where $s \in \mathcal{T}_{str(, con)}^\Sigma$ and $L\subseteq \Sigma^*$ is a formal language,\\
(A2) Length constraints of the form $\ell_1 = \ell_2$ where $\ell_1,\ell_2 \in \mathcal{T}_{arith}^\Sigma$,\\
(A3) Word equations (string-equality constraints) of the form $s_1 = s_2$ where $s_1,s_2 \in \mathcal{T}_{str,con}^\Sigma$.

Formulas in our theories are constructed in general as follows:\\
(F1) Any atom is a well-formed formula,\\
(F2) If $f_1,f_2$ are well-formed formulas then $\neg f_1$ is a well-formed formula and $f_1 \oplus f_2$ is a well-formed formula for each $\oplus \in \{\land,\lor,\implies,\iff\}$.

Note that all formulas are quantifier-free. The semantics associated with these formulas are defined in the natural way: given a substitution for the variables $x_1,x_2,\ldots$ for words in $\Sigma^*$, each string term evaluates to a word in $\Sigma^*$ (possibly as the result of concatenating several smaller words in the case of extended string terms). Each length term is a linear combination of lengths of strings and evaluates to an integer. Atoms of type A1 evaluate to ``true'' if the string term $s$ evaluates to a word in the language $L$ and false otherwise. Atoms of type A2 evaluate to true if the two length terms $\ell_1,\ell_2$ evaluate to the same integer and false otherwise. Atoms of type A3 evaluate to true if the string terms $s_1$ and $s_2$ evaluate to the same word and false otherwise. Finally, Boolean combinations of the form F2 are evaluated in the canonical way.\looseness=-1

The most general logical theory we consider includes all of the above and we consider language membership constraints $s \in L$ where $L$ is a deterministic context-free language, given e.g. as a deterministic  push-down automaton or a context-free grammar. However, we are not just interested in this theory alone, rather we want to consider various sub-theories in order to compare their expressive power and computability-related properties.

We have two ways of restricting expressive power. The first is to restrict the types of terms/atoms we allow, while the second is to restrict the kind of languages we allow in the language membership constraints (atoms of type A1). For the latter, we focus on three main possibilities: regular languages, visibly push-down languages, and deterministic context-free languages. For technical completeness, we can assume that all language constraints are given as automata (NFA, Visibly-PDA, or Deterministic-PDA respectively), however, since we do not focus on precise complexity-related issues, equivalent language descriptors such as grammars could equally be used. In particular, we might use simpler descriptors where convenient to do so and where it is obvious that an equivalent automaton could be constructed. \looseness=-1

We consider all combinations of atom-types A1, A2 and A3, and in each case define versions in which only basic string terms from $\mathcal{T}_{str}^\Sigma$ are allowed and versions in which concatenations of string terms (i.e. terms from $\mathcal{T}_{str,con}^\Sigma$) are allowed. Note that whenever we allow word equations (so, atoms of type A3), we might as well allow concatenations of string terms. If we allow concatenations in word equation terms, then we can model concatenation in all string terms anyway and if we were to restrict equality between string terms to basic string terms only, then we could easily eliminate all string equalities by direct substitution. 

Moreover, we are not going to consider explicitly the case that only length constraints (atoms of type A2) are allowed, since this reduces to the existential fragment of Presburger arithmetic and is therefore not really a string-based logic. With these exclusions, we are left with a total of 20 theories to consider. In fact, since the theories themselves depend on the underlying alphabet $\Sigma$, we have 20 families of theories. As such, it is convenient to introduce a naming convention for these (families of) theories.

If atoms of type A1 are allowed, we add either $\REG$, $\VPL$, or $\CF$ to the name of the theory-family depending on the class of languages permitted: $\REGlang$, $\VPLlang$, or $\CFlang$, respectively. If atoms of type A2 are allowed, we add the abbreviation $\LEN$, separated if necessary by a "$\+$". Likewise, if atoms of type A3 are allowed, we add the abbreviation $\WE$. Finally, if atoms of type A3 are not allowed, but extended string terms are (so we have concatenation but not equality between string terms), then we add the abbreviation $\CON$. Note that $\CON$ is superseded by $\WE$ due to reasons explained above. For example, the most general theory which allows all three atom types (with deterministic context-free languages for atoms of type A1) is denoted by $\WE \+ \CF\+\LEN$. Similarly, $\REG \+ \LEN\+\CON$ describes the theory in which atoms of type A1 (where $L$ is a regular language and $s$ is an extended string term) and A2 are allowed.

For theories allowing $\VPLlang$ membership constraints (i.e. belonging to families of the form $\VPL +\ldots$), we assume a fixed partition of the alphabet $\Sigma$ into the call, return and internal alphabets $\Sigma_c, \Sigma_r, \Sigma_i$. We conclude this section with the following remark.

\begin{remark}\label{remark:basictheories}
Since $\REGlang$ (respectively, $\VPLlang$) is closed under union, intersection and complement, the set of languages  expressible in $\REG$ (respectively, $\VPL$) is exactly $\REGlang$ (respectively, $\VPLlang$). However, the same is not true for $\CF$ and $\CFlang$, since that class is not closed e.g. under intersection. For $\CF$ the expressible languages are exactly the Boolean closure of the deterministic context-free languages. Moreover, it can be inferred from well-known results on word equations (see e.g.~\cite{karhumaki2000,lot:alg}) that the languages expressed by $\WE$ are exactly those expressible by a single word equation in the sense of~\cite{karhumaki2000}.
\end{remark}

\section{Separation and Grouping of Theories}\label{sec:separation}

\begin{figure*}
\hspace*{-10pt} \includegraphics{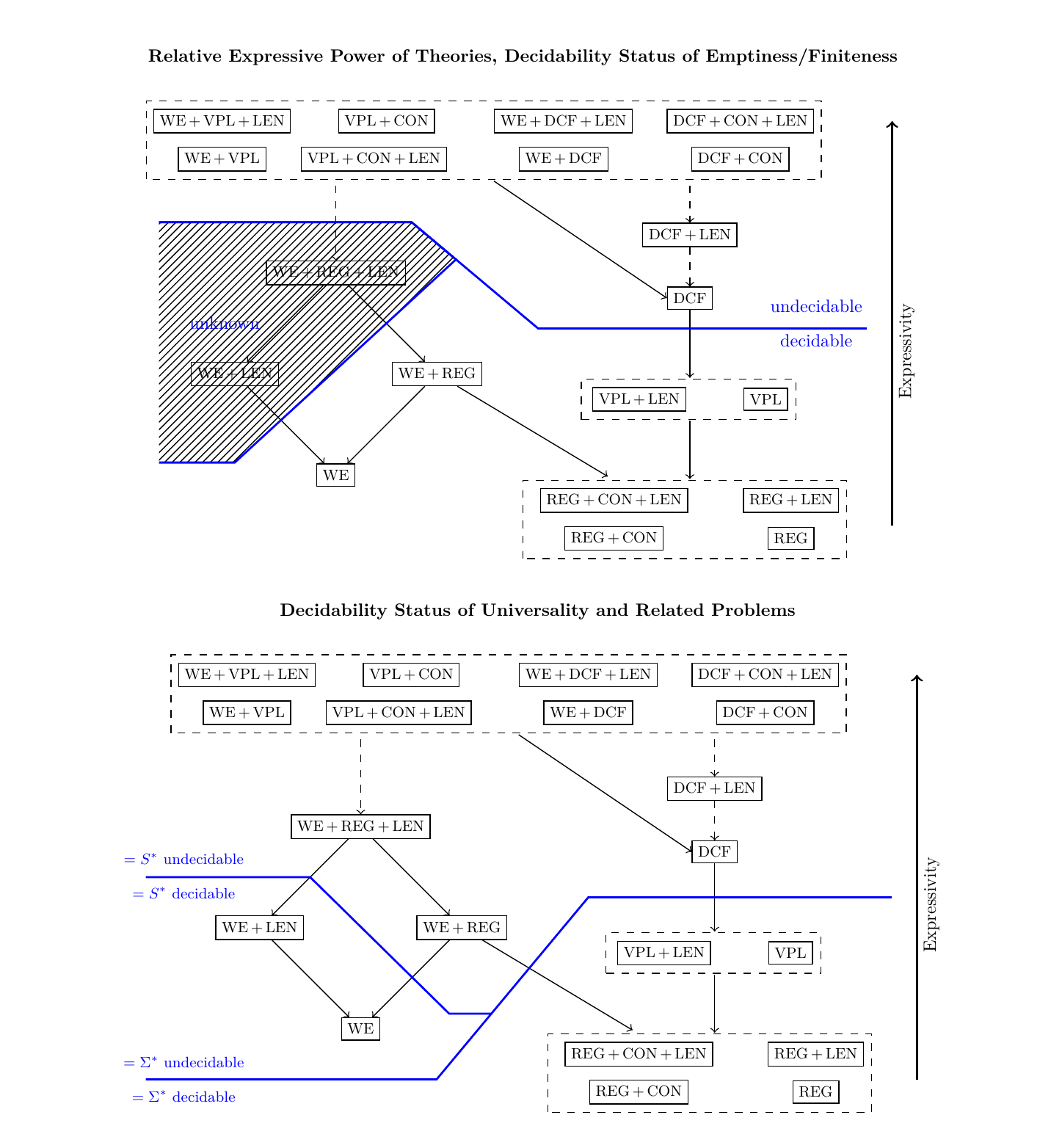}

\caption{Visual representations of all 20 families of string-based logical theories considered. Theory-families are depicted in solid square boxes containing their names (see Section~\ref{sec:theories}). Dashed square-boxes around multiple theory-families show equivalence w.r.t. the class of expressible languages (so equivalence under $\sim$). The arrows between theory-families and their transitive closure represent inclusion w.r.t. the class of expressible languages. 
Solid arrows indicate that the inclusion is known to be strict, while dashed arrows indicate that we do not know whether the inclusion is strict or not. The most expressive group of theories (i.e. those equivalent to $\VPL\+\CON $) are able to express $\RE$. The left figure indicates for which (families of) theories the emptiness and finiteness problems are decidable or undecidable. Our contribution includes undecidability for $\VPL \+ \CON$. This is a particularly interesting case because it is only a slight generalisation of important cases which remain open, namely $\WE \+ \REG \+ \LEN$ and $\WE \+ \LEN$ (shaded). We also show decidability for $\VPL \+ \LEN$, $\REG \+ \CON \+ \LEN$ etc. by showing equivalence to families for which decidability is known ($\REG$ and $\VPL$). The right figure depictes (families) of theories for which the universality ($=\Sigma^*$) and subset-universality ($= S^*$) problems are decidable/undecidable. Equivalence and inclusion (where the two languages might come from different theories) are decidable if and only if both theories fall into cases where universality ($=\Sigma^*$) decidable.
}\label{fig:overviewoftheories}
\end{figure*}

We are interested primarily in whether we can decide properties of a language expressed by a given formula and variable. Therefore, the first thing we consider is the relative expressive power of the various theories defined in the previous section. In particular, we want to understand how the classes of languages which may be expressed by a formula/variable from a given theory relate to each other. To make these comparisons formally, we define the following relation(s) on two logical theories $\mathfrak{T}_1, \mathfrak{T}_2$ whose formulas contain string variables. 

\begin{definition}\label{def:classes}
Let $\mathfrak{T}_1, \mathfrak{T}_2$ be theories whose formulas contain string-variables. We say that $\mathfrak{T}_1 \preceq \mathfrak{T}_2$ if, for every formula $f \in \mathfrak{T}_1$ and every (string) variable $x$ occurring in $f$, there exists a formula $f' \in \mathfrak{T}_2$ and variable $x'$ in $f'$ such that the languages expressed by $x$ in $f$ and $x'$ in $f'$ are identical. Moreover, we say that $\mathfrak{T}_1 \sim \mathfrak{T}_2$ if both $\mathfrak{T}_1 \preceq \mathfrak{T}_2$ and $\mathfrak{T}_2 \preceq \mathfrak{T}_1$ hold. We write $\mathfrak{T}_1 \prec \mathfrak{T}_2$ if  $\mathfrak{T}_1 \preceq \mathfrak{T}_2$ and $\mathfrak{T}_1 \not\sim \mathfrak{T}_2$.
\end{definition}

Hence, $\mathfrak{T}_1 \preceq \mathfrak{T}_2$ if the class of languages expressible in $\mathfrak{T}_1$ is a subset of the class of languages expressible in $\mathfrak{T}_2$, and $\mathfrak{T}_1 \sim \mathfrak{T}_2$ if the two classes are equal.
Note that the relation $\sim$ is an equivalence relation that is a weaker notion of equivalence than being isomorphic. That is, two theories need not be isomorphic to satisfy the equivalence $\sim$.

We extend Definition \ref{def:classes} for the families of theories defined in Section \ref{sec:theories} as follows. Recall that each family contains all the theories consisting of a particular set of formulas, but whose underlying alphabet $\Sigma$ may vary.\looseness=-1

\begin{definition}
Let $\mathfrak{F}_1, \mathfrak{F}_2$ be families of theories as defined in Section \ref{sec:theories}. We say that $\mathfrak{F}_1 \preceq \mathfrak{F}_2$ if, for every theory $\mathfrak{T}_1 \in \mathfrak{F}_1$, there is a theory $\mathfrak{T}_2 \in \mathfrak{F}_2$ such that $\mathfrak{T}_1 \preceq \mathfrak{T}_2$. The relations $\sim$ and $\prec$ are then defined analogously,
\end{definition}

Before moving on, let us make some remarks. It will often be the case that there exist formulas such that the language expressed by a variable $x$ occurring in both formulas is the same, but the sets of satisfying assignments, when considered as a whole, are not identical (see Remark \ref{rem:identicalSolutions} below). This has an important implication for what conclusions we can and cannot draw from a statement of the form e.g. $\mathfrak{T}_1 \sim \mathfrak{T}_2$. E.g., while we will later show that $\REG \sim \REG \+ \LEN$, this does not imply that $\WE \+ \REG \+ \LEN \sim \WE \+ \REG$. Indeed we shall also show explicitly that the latter does not hold.

\begin{remark}\label{rem:identicalSolutions}
Consider the $\LEN$ formula $|x| = 2|y|$ where $x,y$ are string variables. Then the language expressed by $x$ is the set of all even-length words over the underlying alphabet $\Sigma$, and the language expressed by $y$ is simply $\Sigma^*$. Both of these languages are regular, and can be expressed in $\REG$. However, if we were to consider e.g. a $\WE \+ \LEN$ formula $x = yyy \land |x| = 2|y|$, then we cannot replace the condition $|x| = 2|y|$ with constraints based on the aforementioned regular languages. The problem with doing so would be that it allows us to decouple the sets of values for $x$ and $y$ satisfying the length constraint (so we get an $x$, $y$, $x'$, $y'$ such that $|x| = |y'|$ and $|x'| = |y|$ and $x = yyy$ holds, but where $x'$ might be different from $x$ and $y'$ might be different from $y$. \looseness=-1
\end{remark}

In~\cite{karhumaki2000} the authors consider expressibility of languages (and relations) by word equations and show that a language is expressible by $\WE$ if and only if it is expressible by a single word equation. The authors of~\cite{karhumaki2000} also show that, for $\Sigma \supseteq \{a,b,c\}$, the regular language $\{a,b\}^*$ is not expressible by a single word equation, and thus not in $\WE$. The same holds for the language $\{a^nb^n \mid n \in \mathbb{N}_0\}$. Since these languages are clearly expressible in $\WE + \REG$ and $\WE + \LEN$ respectively, we may immediately conclude the following.\looseness=-1  

\begin{theorem}[\cite{karhumaki2000}]
The following hold: $\WE \, \prec \, \WE \+ \REG$ and $\WE \, \prec \, \WE \+ \LEN$. 
\end{theorem}

On the one hand, all languages expressed in our theories are clearly recursively enumerable. On the other hand, in our first main result, we show that, in fact, all recursively enumerable languages can be expressed with only concatenation and $\VPL$-membership.



\begin{theorem}\label{thm:VPLCONequalsRE}
The class of languages expressible in the familiy $\VPL\+ \CON$ is exactly $\RE$. 
\end{theorem}

\begin{proof}
Let $L\subset \Sigma ^*$ be a recursively enumerable language. It is not hard to show that there exists a 1-Tape deterministic Turing machine $M$ accepting $L$, which additionally has the following properties:
\begin{itemize}
\item $M$ has the set of states $Q$, including a final state $q_f$ and an initial state $q_0$, the input alphabet $\Sigma$, the working alphabet $\Gamma $ which includes $\Sigma$ as well as the blank symbol $B$ and a special delimiter-symbol $\$ $. The transition function of $M$ is $\delta: Q\times \Gamma \ra Q\times (\Gamma \setminus \{B\}) \times \{R,L\}$  (where $R$ and $L$ are symbols denoting a left and, respectively, right movement of the tape-head of $M$). 
\item $M$ has a semi-infinite tape (bounded to the left). We assume that the delimiter-symbol $\$ $ marks the left end of the tape. 
\item $M$ accepts the input word or goes in an infinite loop. Moreover, $M$ accepts only after making at least one step, and in the last (i.e., accepting) configuration the tape-head scans the leftmost blank cell of the tape, and the state of $M$ is $q_f$.
\item The delimiter $\$ $ cannot be modified, and $\$ $ cannot be written on any other cell of the tape, than the leftmost one (i.e., at a certain step, $M$ can only write $\$ $ in the cell that already contains $\$ $).
\item In the initial configuration of $M$, the tape-head scans the cell containing the delimiter $\$ $.
\end{itemize}

The configurations of $M$ (i.e., snapshots of the tape of $M$ during the computation) are described by strings $u (q,a) v B^k $, where $a\in \Gamma$, $uv \in \$ (\Gamma\setminus \{B\})^*$ such that $uavB^\omega$ is the the current content of the tape of the machine (where $B^\omega$ means a right-infinite string containing only blanks), $q$ is the current state of the machine, the tape-head scans the cell containing $a$, and $k\geq 1$. Note that $uavB^k$ is a prefix of the content of the tape, read left to right, including one or more $B$ symbols. Note also that the pairs $(q,a)$, for all $q\in Q$ and $a\in \Gamma$, will be part of the call-alphabet for the VPAs we construct in the rest of the proof. At this point, it is very important to remark that there might be more than one string describing the same configuration of the machine.

Now, given $C_1$ and $C_2$ two strings describing configurations of $M$, we say that there is a transition from $C_1$ to $C_2$, denoted $C_1\vdash C_2$, if $C_2$ is a string describing the configuration in which $M$ transitions from the configuration described by $C_1$ and, moreover, $|C_1|=|C_2|$. 

Indeed, as the strings encoding configurations contain an arbitrary large number of blank symbols, if $C$ describes a configuration of $M$ reachable in a finite, greater or equal to one number of steps by $M$ on the input $w$, then there exists a sequence of strings $C_0,\ldots,C_k$ describing configurations of $M$ such that $C_0$ describes the initial configuration of $M$ for the input $w$ and $C_0\vdash C_1\vdash \ldots \vdash C_k=C$ (meaning that $|C_i|=|C_j|$ for all $i,j\in \{0,\ldots,k\}$, as well). Intuitively, $C_0$ already contains all the blanks which will be scanned by the tape-head during the computation of $M$ on $w$, until $C$ is reached. 

A direct consequence of the observation made above is that if $C$ is a string which describes the final configuration of $M$ for the input $w\in L$, then there exists a sequence of strings $C_0,\ldots,C_k$ describing configurations of $M$ such that $C_0$ describes the initial configuration of $M$ for the input $w$ and $C_0\vdash C_1\vdash \ldots \vdash C_k=C$. This also means that $|C_i|=|C_j|$ for all $i,j\in \{0,\ldots,k\}$. So, an accepting computation of $M$ can be described by a sequence of strings encoding configurations, all having the same length. This length equality between the Turing machine's configurations in our setting is a point of the novelty in our proof, and it is crucial for the ''simulation" of its computations by visibly pushdown automata.

We now define an alphabet $\widetilde{\Delta}$ as the triple $(\Delta_c, \Delta_i, \Delta_r)$ of pairwise-disjoint alphabets, which stand for the call, internal and return alphabets, respectively, for the VPAs which we will construct from now on. We define $\Delta_c=\Gamma \cup \{\#,@\}\cup \{(q,a)\mid q\in Q, a\in \Gamma\}$ and $\Delta_r=\{a'\mid a\in \Delta_c\}$; that is, $\Delta_r$ consists in copies of the letters of the alphabet $\Delta_c$. Finally, $\Delta_i=\{ \blacksquare \}$. Let $f:\Delta_c^*\ra \Delta_r^*$ be the antimorphism defined by $f(a)=a'$ for all $a\in \Delta_c$ (here antimorphism means that $f(wu)=f(u)f(w)$ for all $u,w\in \Delta_c^*$).  

Let us now define the language $$L_1=\{@ C_1\#C_2\#\cdots \#C_k \# \blacksquare \#' C'_k\#'\cdots \#'C'_1 @' \}$$ where:
\begin{itemize}
    \item For $i\leq k-1$, $C_i$ is a configuration of $M$ and $C'_i=f(D_i)$, where $C_i\vdash D_i$. In other words, $C'_i$ is the image under $f$ of the string describing the configuration which follows the configuration described by $C_i$ in a computation of $M$. 
    \item $C_k$ describes a final configuration of $M$, and we have that $D_k=C_k$ and $C'_k=f(C_k)$. 
\end{itemize}

We can show that $L_1$ is accepted by a nondeterministic VPA $E$. This VPA functions according to the following algorithm.
\begin{enumerate}
    \item $E$ uses a stack, which is represented as a word and the top of the stack is the rightmost symbol of the respective word.
    \item In the first move, $E$ reads $@ $ and writes $@'$ on the stack.
    \item $E$ computes and writes the strings $D_1\#$, $\ldots$, $D_k \#$ on the stack (by pushing the symbols of these strings in order left to right), while reading the strings  $C_1\#,$ $\ldots$, $\#C_k \#$, respectively. Then $E$ checks whether $f(D_1\# \cdots D_k \#)= \#' C'_k\#'\cdots \#'C'_1$ by iteratively popping a symbol from the top of the stack if and only if it matches the current symbol read on the input tape. $E$ accepts the input if, when the input tape was completely read, the stack is empty (this can be checked using the fact that the last symbol popped must be $@'$).
    \item When reading any of the strings $C_i=u (q,a) v B^k$ (which basically occurs either at the very beginning of the computation or after reading a $\#$ symbol) the automaton $E$ works as follows:
    \begin{itemize}
                \item $E$ non-deterministically guesses if $C_i$ is the final configuration. If $C_i$ is the final configuration, then $E$ reads the symbols $d$ of $C_i$ and pushes $d'$ on the stack, until it reaches a symbol $(q,B)$. If $q$ is not final or $E$ met a $B$ symbol before reaching $(q,B)$, then $E$ goes into an error state. Otherwise, it writes $(q,B)'$ on the stack. It then reads the remaining $B$ symbols and $\#$ and pushes corresponding $B'$ symbols and $\#'$, respectively, in the stack. Then it checks if $\blacksquare$ follows on the input tape. If yes, it simply switches to the part where symbols are popped from the stack. If not, then $E$ goes to an error state. If $C_i$ is not a final configuration, then the following computation is implemented. 
                \item $E$ non-deterministically chooses the transition that $M$ makes in this configuration, and keeps track of this choice in the state.
                \item Assume first that the respective transition is $\delta(q,a)=(q_1,b,L)$. This means that $u$ is non-empty. In this case, $E$ reads the symbols $d\in \Gamma\setminus \{B\}$ of $C_i$ and pushes $d'$ on the stack, until it non-deterministically decides that it has reached the last symbol $c$ of $u$. It then reads $c$ and pushes $(q_1,c)'$ on the stack. The next symbol read on the tape should now be $(q,a)$; otherwise $E$ goes in an error state. If $E$ reads $(q,a)$, it pushes $b'$ on the stack, and it then continues reading the symbols $d$ of $C_i$ and pushing $d'$ on the stack, until it reaches $\#$. It then reads $\#$, pushes $\#'$ on the stack, and moves to the next configuration.
                \item  Assume now that the respective transition is $\delta(q,a)=(q_1,b,R)$. In this case, $E$ reads the symbols $d$ of $C_i$ and pushes $d'$ on the stack, until it reads the symbol $(q,a)$; otherwise, if it goes all the way to the first $B$ symbol without finding $(q,a)$, $E$ goes into an error state. Now, if $E$ reads $(q,a)$,  it pushes $b$ on the stack. Then, $E$ reads the next symbol on the tape. If this symbol is $\#$, then $E$ goes into an error state (intuitively, we cannot compute $D_i$ because it is longer than $C_i$). Otherwise, if this symbol is $d\neq \#$, then $E$ writes $(q_1,d)'$ on the stack. Further, $E$ continues reading the symbols $e$ of $C_i$ and pushing $e'$ on the stack, until it reaches $\#$. It then reads $\#$, pushes $\#'$ on the stack, and moves to the next configuration.
    \end{itemize}
\end{enumerate}
It is not hard to see that $E$ accepts $L_1$.

Intuitively, $E$ accepts those strings consisting in correctly matched (in a palindromic fashion) pairs of consecutive configurations of the Turing machine $M$. Very importantly, there is no connection between different pairs of matching configurations. This type of connection, ultimately leading to a way of expressing the valid computations of $M$, is something that we now need to achieve.

To the end, we define $L_2=\{w\blacksquare f(w)\mid x\in \Delta_c\}$. It is immediate that  $L_2$ can be accepted by a VPA. Similarly, all regular languages can be accepted by VPAs (as $\REGlang\subseteq \VPLlang$).

So, to achieve our goal, we define the following formula $\phi = (x\in \Sigma^*) \land (v\in \{B\}^+) \land (z\in (\Delta_r\setminus \{\#'\})^*) \land (@ (q_0,\$) x v \#y \blacksquare \#' z \# u @' \in L_1) \land (y\blacksquare \# u \in L_2)$. 

We claim that the language expressed by $x$ is $L$. 

Firstly, from $@(q_0,\$) x v \#y \blacksquare \#' z\#' u@' \in L_1$ we get that $$@ (q_0,\$) x v \#y \blacksquare \#' z \#' u@' = @ C_1\#C_2\#...\#C_k \# \blacksquare \#' C'_k\#'....\#'C'_1@',$$ for some configurations $C_1,\ldots, C_k$. Moreover, $C_1$ starts with $(q_0,\$)$ and $x$ is the string contained between $(q_0,\$)$ and the first occurrence of $B $ in $C_1$, so $C_1$ must be an initial configuration of $M$ from which $x$ extracts the input string. Also, $y=C_2\#\cdots C_k\#$. Finally, we obtain that $z=C'_k$ and $ u = C'_{k-1}\#' \cdots \#'C'_1$. 

Secondly, from $y\blacksquare \#' u \in L_2$, we get that $C'_{i-1}=f(C_i)$, for all $i\in \{2,\ldots,k\}$. This means that $C_{i-1}\vdash C_i$, for all $i\in \{2,\ldots,k\}$. Therefore, $C_1\vdash \ldots \vdash C_k$ is an accepting computation of $M$. In conclusion, the word expressed by $x$ is in $L$.

The converse implication is immediate. If $w\in L$, then there exists a sequence of strings $C_0,\ldots,C_k$ (all having the same length) describing configurations of $M$ such that $C_0$ describes an initial configuration of $M$ for the input $w$, $C_k$ describes a final configuration, and $C_0\vdash C_1\vdash \ldots \vdash C_k$. From this, we can easily derive an assignment of the string variables $v,z,y,u$ with $x=w$ for which $\phi$ is satisfiable.

Our claim now follows, and we have shown that any recursively enumerable language can be expressed in $\VPL\+\CON$.
\end{proof}

Consequently, the class of languages expressible in each of $\VPL \+ \CON \+ \LEN$, $\WE \+ \VPL$, $\WE \+ \VPL \+ \LEN$, $\CF \+ \CON$, $\CF \+ \CON \+ \LEN$, $\WE \+ \CF$, and $\WE \+ \CF \+ \LEN$ is the class of recursively enumerable languages $\RE$. Thus, all these theories are equivalent under $\sim$.\looseness=-1

We therefore get a natural hierarchy of theories which extend the syntax and expressive power of $\WE$, $\WE \+ \REG$, and $\WE \+ \LEN$. Next, we show some cases where this hierarchy does not collapse by providing some inexpressibility results.

We need the following technical lemma, which non-trivially extends a similar one from~\cite{karhumaki2000} to accommodate the addition of other types of constraints to word equations.

\begin{lemma}\label{lem:WELENSEP}\label{lem:WEREGSEP}
Let $\mathfrak{F}$ be a property on words defining a synchronising factorisation. 

Let $\mathfrak{T}$ be a theory belonging to $\WE \+ \LEN$ or to $\WE \+ \REG$, and let $f$ be a formula from $\mathfrak{T}$ and $x$ a variable occurring in $f$. Suppose that $w$ belongs to the language $L$ expressed by $x$ in $f$ and let $w = w_1 w_2 \ldots w_n$ be its $\mathfrak{F}$-factorisation. Then, the following holds:\\
(1) if $\mathfrak{T}$ belongs to $\WE \+ \LEN$, there exists $d \in \mathbb{N}$ such that if the number of distinct factors $w_i$ is greater than $d$, then there exists at least one $i$ such that, for every word $u$ with $|u| = |w_i|$, the word $w'$ obtained by replacing each occurrence of $w_i$ in $w$ with $u$ also belongs to $L$; \\
(2) if $\mathfrak{T}$ belongs to $\WE \+ \REG$, there exist $d,e \in \mathbb{N}$ 
such that if the number of distinct factors $w_i$ with $|w_i| > e$ is greater than $d$, then there exists at least one $i$ and a word $u$ with $|u| < |w_i|$ such that the word $w'$ obtained by replacing each occurrence of $w_i$ in $w$ with $u$ also belongs to~$L$.\looseness=-1
\end{lemma}

\begin{proof}$ $

{\bf Proof of Part (1):}
Firstly, we write the formula $f$ in Disjunctive Normal Form (DNF). We can then separate each clause into conjunctions of word equation literals and length constraint literals. By canonical constructions (see e.g.~\cite{karhumaki2000}), we can replace the word equation literals with a single word equation atom (i.e. a positive occurrence of a word equation) without affecting the language expressed by $x$. Thus the language expressed by $x$ is a finite union of languages, each given by the solutions to a single word equation which adhere to the length constraint literals (projected onto the variable $x$).

Suppose that $w \in L$. Then we must have that $w$ is the value of $x$ of some solution to one of the word equations corresponding to one of the DNF clauses as described above.

The proof of Theorem~16 in~\cite{karhumaki2000} directly establishes the fact that under the conditions of the lemma, we may replace all occurrences of some $w_i$ consistently with any other word $u$ and still have a solution to the word equation. By ensuring that we swap $w_i$ for a word $u$ of the same length, we ensure that the new solution satisfies the length constraints whenever the old one does. It follows that the new word also belongs to $L$ so the statement of the lemma holds.

{\bf Proof of Part (2):} This is the more complicated part of this proof. Just like in the previous case, we start by writing the formula $f$ in Disjunctive Normal Form (DNF). We can then separate each clause into conjunctions of word equation literals and regular language membership literals. By closure properties of regular languages, we can remove negations of regular language membership constraints and by canonical constructions (see e.g.~\cite{karhumaki2000}), we can replace the word equation literals with a single word equation atom (i.e. a positive occurrence of a word equation) without affecting the language expressed by $x$. Thus the language expressed by $x$ is a finite union of languages, each given by the solutions to a single word equation which adhere to the (positive) regular language membership constraints (and projected onto the variable $x$).

Suppose that $w \in L$ satisfies the conditions of the lemma. Then we must have that $w$ is the value of $x$ of some solution to one of the word equations corresponding to one of the DNF clauses as described above. Let $x_1,x_2,\ldots, x_m$ be the variables occurring in that clause, and for each $j, 1\leq j \leq m$, let
$A_j$ be the minimal DFAs describing the intersection of all regular language membership constraints acting on the variable $x_j$ and let $Q_j$ be the set of states of $A_j$. For $1\leq j \leq m$, let $w_j = w_{j,1} w_{j,2} \ldots w_{j,k_j}$ be the $\mathfrak{F}$-factorisation of the value of $x_j$ in the aforementioned solution from which we get $w$.

Now, the proof of Theorem~16 in~\cite{karhumaki2000} directly establishes the fact that under the conditions of the lemma, we may replace all occurrences of some $w_i$ with $|w_i| > e$ consistently with any other word $u$ and still have a solution to the word equation. We need to ensure that this replacement also respects the regular language membership constraints, and this is the novel part of this proof. In particular, while in part (1) a similar replacement was done while preserving the length of the replaced strings, here we need to change the length of the solution. 

To see how this can be achieved, for each position $p$ of $w_i$ (that is for each number $p$ with $1\leq p \leq |w_i|$), and for each $j$, $1\leq j \leq m$ we associate a function $g_{p,j} : Q_j \to Q_j$ such that $g_{p,j}(q) = q'$ if the automaton $A_j$ finishes in state $q'$ when reading the first $p$ letters of $w_i$ starting in state $q$. 

Note that if $g_{p,j}(q) = g_{p',j}(q)$, and when the automaton reads $w_j$ it begins reading an occurrence of $w_i$ in state $q$, then the factor of that occurrence of $w_i$ between positions $p$ and $p'$ can be removed without affecting acceptance. Thus, if $g_{p,j} = g_{p',j}$, then the factor between positions $p$ and $p'$ can be removed from all occurrences of $w_i$ in $w_j$ without affecting the acceptance w.r.t. $A_j$. Finally, if there exist $p,p'$ such that $g_{p,j} = g_{p'j}$ for all $j, 1\leq j \leq m$, then the factor between positions $p$ and $p'$ of $w_i$ can be removed from all occurrences of $w_i$ in all words $w_j$ without affecting the satisfaction of any of the regular membership constraints. 

Once the formula $f$ is fixed, the cardinalities of the sets of states $Q_j$ are also fixed. Moreover, the number of possible values for $j \leq m$ is fixed. Thus the number of distinct functions $g_{p,j}$ is bounded and so is the number of combinations of functions $(g_{p,j})_{1 \leq j \leq m}$. Let $e$ be the number of possible combinations of these functions. Then, if $|w_i| > e$, we must necessarily be able to find $p,p'$ with $p \not= p'$ such that $g_{p,j} = g_{p',j}$ for all $j$. By taking $u$ to be the word obtained by removing the factor between positions $p$ and $p'$ in $w_i$, we therefore obtain $|u| < |w_i|$ and still satisfy all regular language membership constraints as required.
\end{proof}

Using Lemma \ref{lem:WELENSEP}(1), we can show that there is a language expressible in $\WE \+ \REG$ (thus, in $\WE \+ \REG \+ \LEN$ as well) which is not expressible in $\WE \+ \LEN$. Similarly, using Lemma \ref{lem:WEREGSEP}(2) we can show that there is a language expressible in $\WE \+ \LEN$ (and thus in $\WE \+ \REG \+ \LEN$) which is not expressible in $\WE \+ \REG$.

\begin{lemma}\label{lem:subalphabetexample}\label{lem:WELENnotWEREG}
Let $a,b,c$ be distinct letters. (1) The language $L = \{ v c \mid v \in \{a,b\}^*\}$ is expressible in $\WE \+ \REG$ but not expressible in  $\WE \+ \LEN$. (2) The language $L = \{ u c v \mid u,v \in \{a,b,c\}^* \land |u| = |v|\}$ is expressible in $\WE \+ \LEN$ but not expressible in $\WE \+ \REG$.
\end{lemma}

\begin{proof}$ $

{\bf Proof of Part (1):}
Suppose to the contrary that $L$ is expressible in a theory from $\WE \+ \LEN$. Let $\Sigma \supseteq \{a,b,c\}$ be the underlying alphabet of that theory. Let $f$ be a formula from that theory and $x$ a variable in $f$ such that $L$ is the language expressed by $x$ in $f$. Let $d$ be the constant from Lemma~\ref{lem:WELENSEP}.

Now, clearly $w= aba^2ba^3b \ldots a^d b c$ belongs to the language $L$. Moreover, recall that the factorisation $\mathfrak{F}$ of words into blocks of single letters is a synchronising factorisation. Now, the $\mathfrak{F}$-factorisation of $w$ consists of at least $d + 2$ distinct factors (namely $a^i$ for $1\leq i \leq d$, $b$ and $c$), so by Lemma~\ref{lem:WELENSEP} we may swap all occurrences of at least one of these factors for any word of the same length and we will get another word in the language $L$. However, whichever block of letters we swap, we can always choose a word of equal length which violates membership in $L$. If we have a block consisting of $a$'s or of a $b$, we can swap it for a word in $c*$ and the resulting word will not be in $L$. Likewise, if we swap the block $c$, we could swap it e.g. for an $a$ and again the resulting word will not be in $L$. In all cases we get a contradiction, so $L$ cannot be expressible in $\mathfrak{T}$ as required.

{\bf Proof of Part (2).}
The fact that $L$ is expressible in $\WE \+ \LEN$ is straightforward. To see that it is not expressible in $\WE \+ \REG$, suppose to the contrary that it is. Then there exists a theory $\mathfrak{T}$ from $\WE \+ \REG$ containing a formula $f$ such that $L$ is expressed by a variable $x$ in $f$. Let $d,e$ be the constants from Lemma~\ref{lem:WEREGSEP}. Then there exist pairwise distinct numbers $p_1,p_2, \ldots, p_d, q_1,q_2,\ldots, q_d > e$ such that $\sum p_i = \sum q_i$ and thus such that $w = a^{p_1}b^{p_1} a^{p_2}b^{p_2} \ldots a^{p_d} b^{p_d} c a^{q_1} b^{q_1} a^{q_2} b^{q_2} \ldots a^{q_d} b^{q_d}$ belongs to $w$.

Recall that the factorisation $\mathfrak{F}$ of words into blocks of single letters is a synchronising factorisation. Now, the $\mathfrak{F}$-factorisation of $w$ consists of $4d $ distinct factors having length greater than $e$ (namely $a^{p+i}, b^{p_i}, a^{q_i}, b^{q_i}$ for $1\leq i \leq d$), so by Lemma~\ref{lem:WEREGSEP} we may swap all occurrences of at least one block of letters for a strictly shorter block. However, since each block of letters occurs only once (by the fact that the $p_i$s and $q_i$s are all pairwise distinct), this would result in a word for which one side of the central $c$ is shorter than the other, and thus not belonging to $L$, a contradiction. Thus $L$ cannot be expressible as required.
\end{proof}




Summarising, we get that $\WE \+ \LEN, \WE \+ \REG \,\prec \, \WE \+ \REG \+ \LEN$, while the classes of languages expressible in $\WE \+ \LEN$ and $\WE \+ \REG$ are incomparable.

Next, we turn our attention to the remaining  theories which do not extend the expressive power of word equations. Since we have already seen that concatenation together with visibly pushdown (or deterministic context-free) membership constraints is enough to model recursively enumerable languages, and therefore word equations, the remaining theories consist of language membership without concatenation (but possibly with length constraints) and all combinations consisting of regular language membership constraints without word equations (so including either concatenation, length constraints, both, or neither). 

In the following lemma, we state another important result. For this, let ${\mathcal C}$ be a class of formal languages which contains $\REGlang$, is contained in $\GCFlang$, and is effectively closed under intersection and complement. We assume that the languages of ${\mathcal C}$ are specified by an accepting or generating mechanism which allows the construction of a context-free grammar generating that language. Let ${\mathcal C}_t$ be the theory defined as in Section~\ref{sec:preliminaries} which allows only language membership predicates (of type A1) for the class of languages ${\mathcal C}$. Let ${\mathcal C}_t\+\LEN$ be the theory which also allows length constraints. In this framework, the following holds.
\begin{lemma}\label{lemma:cftheories}
${\mathcal C}_t \+ \LEN \sim {\mathcal C}_t$.
\end{lemma}

\begin{proof}
Let $f$ be a formula in ${\mathcal C}_t \+ \LEN$ and let $\alpha_1,\ldots,\alpha_r$ be the atoms appearing in this formula. We can compute all the assignments of truth values to the atoms $\alpha_1,\ldots,\alpha_r$ which make $f$ true. That is, each such assignment $\sigma$ simply maps each $\alpha_i$ to a true or false, such that at the end $f$ evaluates to true according to the values assigned by $\sigma$ to its atoms. Now, for each assignment $\sigma$ as above, we construct a formula $f_\sigma$ as the conjunction of the formulae $\beta_\ell$, for $\ell\in\{1,\ldots,r\}$, where $\beta_\ell=\alpha_\ell$ if $\alpha_\ell$ is true in $\sigma$ and $\beta_\ell=\neg \alpha_\ell$, otherwise.  Moreover, in $f_\sigma$ we replace each language constraint $\neg{(x\in L)}$ by the constraint ${x\in L^C}$, where $L^C$ is the complement of the language $L$. Let $f'$ be the disjunction of all the fomulae $f_{\sigma}$. It is not hard to see that $f'$ is equivalent to $f$, i.e., the assignments of the variables which satisfy $f$ are exactly the same as the assignments of the variables which satisfy $f'$. 

Consider now a formula $f_\sigma$. Let $f_1$ be the sub-formula of $f_\sigma$ consisting in all the arithmetic constraints. 

Let $x$ be a string variable occurring in $f_\sigma$. Let $x\in L^x_i$, for $i\in \{1,\ldots,k_x\}$ and some positive integer $k_x$, be all the language membership constraints involving the variable $x$ occurring in $f_\sigma$. According to our hypotheses regarding the class of languages ${\mathcal C}$, we can compute a representation of the language $L_x=\cap_{i=1,k_x} L^x_i$, and the context-free grammar generating $L_x$.

By Parikh's theorem~\cite{Parikh1966}, we can compute a regular language $R_x$ (and a DFA $A_x$ accepting it) such that the set of Parikh vectors of the words in $R_x$ is equal to the set of Parikh vectors of the words in $L_x$. In particular, this means that $L_x$ contains a word of length $\ell$ if and only if $R_x$ contains a word of length $\ell$, for all $\ell\geq 0$.

Further, let $B_x$ be the unary NFA with $m_x$ states obtained by re-labelling all transitions in $A_x$ with a single letter $a$. It is clear that the paths of $A_x$ correspond bijectively to the paths of $B_x$. Further, we compute the NFA $B'_x$, the Chrobak normal form of the unary automata $B_x$ (see~\cite{Chrobak86,Gawrychowski11}).
As such, $B'_x$ consists of a path of length $O(m_x^2)$ ending in a state $q$, followed by a single nondeterministic choice from $q$ to a set of disjoint cycles of lengths $c_1, c_2, \ldots, c_h$, with $c_i \leq m_x$ for all $i\leq h$. It follows that there exists a finite set of arithmetic progressions which describe exactly the length of accepted words for $B'_x$. Let $\ell^x_{p,i} + \ell^x_{c,i}\nat $, for $i\in \{1,\ldots,d_x\}$, be these progressions corresponding to $B'_x$. Note that for each of these lengths (and, accordingly, for each of the accepting paths of $B'_x$), there exists a word of corresponding length accepted by $B_x$, so there also exists a word of this length in $R_x$, and, ultimately, a word of this length in $L_x$. 

Consequently, for each string variable $x$ we now get a series of new arithmetic constraints on the length of each variable $x$: in a satisfying assignment of $x$, we must have that $|x|$ has the form $\ell^x_{p,i}+\alpha^x_{i} \ell^x_{c,i}$, for some $i\leq d_x$; here $\alpha^x_{i}$ is a positive integer variable. So, we define a new formula $f_{x}$ of length constraints: $f_x=\lor_{i\leq d_x }|x|=\ell^x_{p,i}+\alpha^x_{i} \ell^x_{c,i}$. 

We now define $f_{ar} = f_1 \land (\land_{x\mbox{ string variable}} f_x)$. 

At this point, it is worth noting that a satisfying assignment for $f_{ar}$, which assigns to $|x|$ the integer $\ell_x$, induces an assignment for the lengths of all the string variables occurring in $f_\sigma$. Moreover, due to the construction of the formulae $f_x$, we have that for each such variable $x$ there exists (at least) an accepting path of length $\ell_x$ in $B'_x$. As explained above, this means that for each such variable $x$ there exists (at least) a word of length $\ell_x$ in $L_x$. Thus, a satisfying assignment for $f_{ar}$ induces (at least) one assignment of both the string and the integer variables occurring in $f_\sigma$ which satisfies $f_\sigma$. 

Further, by standard methods, the arithmetic formula $f_{ar}$ can be transformed in a linear system $A\mathbf{y}\geq \mathbf{b}$, where $A$ is a matrix of integers, $\mathbf{y}$ is the vector of integer variables (including the lengths of string variables), and $\mathbf{b}$ is a vector of integers. By the results in~\cite{ChubarovV05}, we get that for a system of inequalities $A\mathbf{y}\geq \mathbf{b}$ we can compute two finite sets of vectors, $H_0$ and $H_1$, such that each integral solution $v$ of the system (i.e., assignment of the ${\mathbf{y}}\gets v$ of the integer and length variables which satisfies $A\mathbf{y}\geq \mathbf{b}$) can be expressed as the sum $v_0 + v_1$ where $v_0$ is a linear combination (with integral non-negative coefficients) of vectors from $H_0$, and $v_1$ is a vector from $H_1$. Therefore, the length $|x|$ of each string variable $x$ has the form $\lambda_{1} a_{x,1} + ... + \lambda_{d} a_{x,d} + b_x$, where $d$ is the number of vectors in $H_0$, the integers $a_{x,j}$ are determined by the vectors of $H_0$, and $b_x$ is determined by the vectors of $H_1$. The parameters $\lambda_{j}$ are positive integers. Moreover, we can assume that all $a_{x,j}$ are positive; otherwise, if some $a_{x,t}$ would be negative, it is enough to note that when $\lambda_{t}$ grows to infinity, and all other $\lambda_{x,j}$ are set to $1$, the length $|x|$ would become negative, a contradiction. 

Now, it is not hard to note 
that the set of strings $w$ whose length has the form $\lambda_{1} a_{x,1} + ... + \lambda_{d} a_{x,d} + b_x$, for any fixed positive integers $a_{x,j}$ and integer $b_x$, is a regular language $D_x$. 

Let us now define $g_\sigma = \land_{x\mbox{ string variable}} (x\in L_x\cap D_x)$. We analyse the sets of strings which can be expressed by $x$ in $f_\sigma$ and $g_\sigma$, respectively. 

Let $w$ be a string which can be expressed by a string variable $x$ in $f_\sigma$. Then, by the construction of the set $D_x$, it is clear that $w\in D_x$. Also, there exists an assignment, which satisfies $f_\sigma$, where $x$ is mapped to $w$ and, for other each variable $y$, the string assigned to $y$ is in $D_y$ (simply because the formula $f_{ar}$ is satisfied in a satisfying assignment of $f_\sigma$). Therefore, the respective assignment for the variables in $f_\sigma$ is also a satisfying assignment for $g_\sigma$, and $w$ can be expressed by $x$ also in $g_\sigma$. 

For the converse, let $u_x$ be a string which can be expressed by $x$ in $g_\sigma$, and let $\ell_x=|u_x|$. This means that $u_x\in D_x\cap L_x$. By definition, $u_x$ is part of at least one satisfying assignment for $g_\sigma$. Note, however, that this initial assignment is not necessarily satisfying $f_{ar}$: it simply consists in an assignment $y=u'_y\in D_y\cap L_y$ for each variable $y$ of $g_\sigma$ other than $x$, and the lengths of these strings are not synchronized in such a way that they form (together with $\ell_x$) a solution for $f_{ar}$. However, because $u_x\in D_x$, we obtain that $f_{ar}$ has a solution with $|x|=\ell_x$. So, there exists an assignment $|y|=\ell_y$, for all the other string variables $y$ occurring in $g_\sigma$, which, together with $\ell_x$, satisfies $f_{ar}$. By the explanations we gave during the construction, it follows that there exists for each variable $y\neq x$ a string $u_y$ (not necessarily the same from the satisfying assignment of $g_\sigma$ considered above) such that the numbers $\ell_y=|u_y|$, for $y$ such that $y\neq x$, and $\ell_x$ are part of a satisfying assignment for $f_{ar}$. Clearly, $u_z\in L_z\cap D_z$, for all string variables $z$ (including $x$). Note that the assignment  $z=u_z$, for all string variables $z$ occurring in $g_\sigma$, is still a satisfying assignment of $g_\sigma$, although not necessarily the one we started with. Thus, as $u_z\in L_z$, for all string variables $z$ of $f_\sigma$, and the tuple defined by the numbers $\ell_z$ (together with an assignment of the integer variables) satisfies $f_1$, it follows that the assignment $z=u_z$, for all string variables $z$ occurring in $f_\sigma$, is part of a satisfying assignment for $f_\sigma$. In conclusion, if $u_x$ is a string expressed by $x$ in $g_\sigma$, we can construct a satisfying assignment for the variables of $g_\sigma$, which is part of a satisfying assignment of the variables of $f_\sigma$. Therefore, $u_x$ can be expressed by $x$ in $f_\sigma$ too. 

Now, it is easy to see that the following statements are equivalent:
\begin{itemize}
    \item $w$ is a string which can be expressed by $x$ in $f$; 
    \item $w$ is a string which can be expressed by $x$ in $f'$; 
    \item there exists an assignment $\sigma$ of the atoms of $f'$ which satisfies $f'$ and $w$ is a string which can be expressed by $x$ in $f_\sigma $; 
    \item there exists an assignment $\sigma$ of the atoms of $f'$ which satisfies $f'$ and $w$ is a string which can be expressed by $x$ in $g_\sigma $;
    \item $w$ is a string which can be expressed by $x$ in $g'$, the disjunction of all the formulae $g_{\sigma}$.
\end{itemize}
Because the class of languages ${\mathcal C}$ includes all the regular languages and is effectively closed under intersection and complementation, we obtain that $g'$ is a ${\mathcal C}_t$-formula. This concludes our proof.
\end{proof}

As $\VPLlang$ is a class which fulfills the properties of the class ${\mathcal{C}}$ from the above lemma, and is strictly included in $\RE$, we immediately get the first claim of the following theorem. The second claim can also be shown with some additional effort. 
\begin{theorem}\label{the:separatinglengthVPL}\label{the:regulartheories}
(1) $\VPL \, \sim \, \VPL \+ \LEN  \prec \, \VPL\+ \CON $.\\
(2) $\REG \, \sim \, \REG \+ \LEN \, \sim \,  \REG \+ \LEN \+ \CON$.
\end{theorem}
\begin{proof}$ $

{\bf Proof of Part (1):}
This is straightforward, as $\VPL$ fulfils the requirements of Lemma~\ref{lemma:cftheories}.

{\bf Proof of Part (2):}
Firstly, we need to introduce a notation. 

If $A=(Q,q_0,F,\delta)$ is a finite automaton with the set of states $Q$, initial state $q_0$, set of final states $F$, and transition function $\delta$, and $q_1,q_2\in Q$, then $A_{q_1,q_2}$ is the finite automaton $A_{q_1,q_2}=(Q,q_1,\{q_2\},\delta)$. In other words, $A_{q_1,q_2}$ is a finite automaton with the same underlying graph as $A$, but with a different initial state ($q_1$ instead of $q_0$) and a different set of final states ($\{q_2\}$ instead of $F$).

We now move to the actual proof. $\REG \, \sim \, \REG \+ \LEN$ follows from Lemma~\ref{lemma:cftheories}, as $\REGlang$ clearly fulfils all properties of the class ${\mathcal C}$ from the respective lemma. 

Let us now consider a formula $f$ from $\REG \+ \LEN \+ \CON$. We construct a formula $f''$ from $\REG \+ \LEN$ which is equivalent to $f$. 

Let $\alpha_1,\ldots,\alpha_r$ be the atoms appearing in $f$. Similarly to Lemma~\ref{lemma:cftheories}, we can compute all the assignments of truth values to the atoms $\alpha_1,\ldots,\alpha_r$ which make $f$ true. For each such assignment $\sigma$ we construct a formula $f_\sigma$ as the conjunction of the formulae $\beta_\ell$, for $\ell\in\{1,\ldots,r\}$, where $\beta_\ell=\alpha_\ell$ if $\alpha_\ell$ is true under $\sigma$ and $\beta_\ell=\neg \alpha_\ell$, otherwise. Moreover, in $f_\sigma$ we replace each language constraint $\neg{(\beta \in L)}$ by the constraint ${\beta \in L^C}$, where $L^C$ is the complement of the language $L$. Let $f'$ be the disjunction of all the formulae $f_{\sigma}$. It is clear that $f'$ is equivalent to $f$, i.e., the satisfying assignments (of the variables) for $f$ are exactly the same as the satisfying assignments (of the variables) for $f'$. 

Consider now a conjunction $g$ of language membership atoms from $\REG \+ \CON$ and arithmetic atoms. We define a {\em consistent language assignment of $g$} as a mapping $\pi$ which leaves every arithmetic atom unchanged and maps every membership atom $\beta=w_0x_1w_1\cdots x_kw_k \in L$ of $g$, where $x_i$ is a string variable for $i\in \{1,\ldots,k\}$, $w_i$ is a constant string for $i\in \{0,\ldots,k\}$, and $L$ is a regular language accepted by a DFA $A=(Q,q_0,F,\delta)$, to the formula $\pi(\beta) = \land_{i=1,k} (x_i\in L(A_{q_{2i-1},q_{2i}}))$, where
\begin{itemize}
    \item $q_1=\delta(q_0,w_0)$,
    \item $q_{2i+1} = \delta(q_{2i},w_i)$, for $i=1,k-1$,
    \item $\delta(q_{2k},w_{k})\in F$.
\end{itemize}
It is clear that the formula $g_\pi$, obtained by replacing each atom $\beta$ of $g$ by $\pi(\beta)$, is now a formula over $\REG + \LEN$. Moreover, it is immediate that the formula $rl(g)$ defined as the disjunction of the formulae $g_\pi$, for all consistent language assignments $\pi$ of $g$, is a formula over $\REG\+\LEN$. Clearly, $rl(g)$ is equivalent to $g$, i.e., the satisfying assignments for $g$ are exactly the same as the satisfying assignments for $rl(g)$. 

Coming back to the formula $f'$ which we have constructed above starting from $f$, we further define a formula $f''$ as the disjunction of all the formulae $rl(f_{\sigma})$. It is clear that $f''$ is a formula from $\REG\+\LEN$, and that $f''$ is equivalent to $f$, i.e., the satisfying assignments for $f$ are exactly the same as the satisfying assignments for $f''$. 

This shows that $\REG\+\CON\+\LEN\sim\REG\+\LEN$. 
\end{proof}

Recall that the languages expressible in $\REG$ (and in $\REG\+\LEN$ and $\REG\+\CON\+\LEN$) and $\VPL$ (as well as $\VPL\+\LEN$) are exactly the classes $\REGlang$ and, respectively, $\VPLlang$, and for each formula in one of these theories we can effectively construct a corresponding automaton accepting the language expressed by a given variable. See Remark \ref{rem:constructiveequivalence} below.\looseness=-1

\begin{remark}\label{rem:constructiveequivalence}
In fact, for a formula $f$ in the theory $\VPL\+\LEN$ (which includes the theories $\REG$, $\REG \+ \LEN \+ \CON$, and $\VPL$) we can effectively construct a formula $g'$ which is a disjunction of conjunctions $g_\sigma$ involving at most one membership predicate $x\in S_x$ per variable, where each language $S_x$ is in $\VPLlang$. We can remove from $g'$ the conjunctions $g_\sigma$ which contain at least one membership predicate $x\in S_x$ with $S_x=\emptyset$. Now, it is easy to see that for the language expressed by $x$ is exactly the union of the languages $S_x$ for all membership predicates $x\in S_x$ occurring in $g'$. Therefore, this language is in the class $\VPLlang$, and we can effectively compute an automaton accepting it. Therefore, we can easily conclude that for two given formulae $f$ and $\phi$ from $\VPL \+ \LEN$ and a variable $x$ occurring in $f$ and a variable $\phi $ occurring in $\phi$, we can decide whether the language expressed by $x$ is the same as (respectively, included in) the language expressed by $y$. 
\end{remark}

Let us now consider the theory $\CF$. The result of Lemma~\ref{lemma:cftheories} does not apply in this case, as the class of languages $\CFlang$ is not closed under intersection. In fact, for Lemma~\ref{lemma:cftheories} to work, it would be enough to have that if $L$ is a finite intersection of languages from the class ${\mathcal C}$ then the set $S=\{|w|\mid w\in L\}$ is semi-linear. However, this still does not hold for $\CFlang$. See Example \ref{exampleDCF} below.

\begin{example}\label{exampleDCF} Let $U_1 = \{a^n b^{2n} \mid n \geq 1\}$ and $L_1 = U_1^+$.
Let $U_2 = \{b^n a^n \mid n \geq 1\}$ and $L_2 = a U_2^+ b^+$. It is clear that $L_1$ and $L_2$ are in $\CF$. 
Let $L=L_1\cap L_2$. It is not hard to observe that $ L= \{ a b^2 a^2 b^4 a^4 b^8 \cdots a^{2^k} b^{2^{k+1}} \mid k \geq 1\}$. Further, let $S=\{|w|\mid w\in L\}$. We have that $S=\{2^{k+2} + 2^{k+1} - 3\mid k\geq 1\}$. Clearly, $S$ is not a semi-linear set (and it is not a deterministic context-free language either).
\end{example}

We now show an additional lemma.
\begin{lemma}\label{lem:exampleCFsep}
$L = \{ w c w \mid w \in \{a,b\}^* \}$ is expressible in $\WE\+ \REG$ and not in $\CF$.
\end{lemma}

\begin{proof}
Recall that the languages expressible in a theory belonging to the family $\CF$ are exactly the closure of the deterministic context-free languages under the Boolean operations: intersection, complement and union. It was shown in~\cite{wotschke1973} that this closure is strictly contained in the intersection-closure of the context-free languages. Moreover, it was also shown in~\cite{wotschke1973} that $L$ in question does not belong to the intersection closure of context-free languages. Thus, we may conclude that it is not expressible in any theory belonging to $\CF$.
\end{proof}

By Theorem~\ref{thm:VPLCONequalsRE} and the existence $\RE$-languages which are not expressible in $\CF$ (see~\cite{wotschke1973}, as well as Lemma \ref{lem:exampleCFsep} or Example \ref{exampleDCF}), we may infer the following relations: 
$\CF \preceq \CF \+ \LEN \preceq \CF + \CON$ and $\CF \prec \CF + \CON $.   
%
%
%
This also shows that at least one of the relations $\CF \preceq \CF \+ \LEN$ and $\CF \+ \LEN \preceq \CF + \CON$ is strict. In fact, there are some indications (see Remark \ref{rem:CF-LEN} below) that the separation might occur between $\CF \+ \LEN$ and $\CF \+ \CON$.

\begin{remark}\label{rem:CF-LEN}
We observe that the language $L=\{wcw\mid w\in \{a,b\}^*\}$, which is expressible in $\CF \+ \CON$ and not $\CF$, is not expressible by a restricted set of formulas in $\CF \+ \LEN$.

Assume that there exists a formula $\phi$ in $\CF\+\LEN$, over the variables $x,y_1,\ldots,y_k$, such that $\phi$ is a conjunction of language membership atoms and arithmetic atoms, and the language expressed by $x$ is $L$. In this case, for each odd number $2k+1$ there exists $(w_x,w_{y_1},\ldots,w_{y_k})$, an assignment of the variables $(x,y_1,\ldots,y_k)$ which satisfies $\phi$, with $|w_x|=2k+1$. Also, if there is an assignment $w'_x$ of $x$ which satisfies all the language membership atoms involving $x$ and also $|w'_x|=2k+1$, then it is not hard to see that $(w'_x,w_{y_1},\ldots,w_{y_k})$ is also a satisfying assignment for $\phi$ (as the membership atoms involving various variables are independent). So, we can now observe that, in this framework, $L$ would be expressed by $x$ in a conjunction of membership atoms: the ones describing $x$ in $\phi$ and a new one, stating that the length of $x$ is odd. This new formula would clearly be in $\CF$, which leads to a contradiction. A similar argument would also work for slightly more complicated formulae $\phi$, but seems to need a nontrivial extension to cover the entire theory $\CF\+\LEN$. However, based on this observation, we conjecture that the relation $\CF \+ \LEN \preceq \CF + \CON$ is, in fact, strict.
\end{remark}



As said above, $\REG\+\LEN \+ \CON$ and $ \VPL \+ \LEN$ express exactly the classes of regular languages and VPL languages, respectively. Since the regular languages are a strict subset of the VPL languages, which in turn are a strict subset of the deterministic context-free languages, we may conclude the following strict inclusions in terms of expressibility:
%
$\REG \+ \LEN \+ \CON \, \prec \,\VPL \+ \LEN \, \prec \, \CF$. 

Note that there are languages expressible in $\WE$ (such as $\{xx \mid x \in \Sigma^*\}$) which are not regular nor visibly pushdown, and thus not expressible in $\REG$ or $\VPL$ or theories with equivalent expressibility. So, 
$\REG \prec \WE \+ \REG$ holds.
Moreover, we have already seen examples of regular languages which are not expressible in $\WE$ or $\WE \+ \LEN$. 
In this context, we leave open the following particularly interesting problem:
\begin{openproblem}
Which recursively enumerable languages (if any) can be shown not to be expressible by a formula from $\WE \+ \REG \+ \LEN$? 
\end{openproblem}



Based on the previous results, we can now also discuss the emptiness problem, and the closely related finiteness problem. This is particularly interesting since emptiness for a language expressed by a formula $f$ and variable $x$ corresponds exactly to the satisfiability problem for $f$.
Based on existing literature \cite{Alur2004,HopcroftUllman,schulz1990,lot:alg}, it is not hard to show that emptiness and finiteness are decidable for $\VPL$ and $\WE\+\REG$ but undecidable for $\CF$.

On the other hand, two cases where it seems particularly difficult to settle the decidability status of the satisfiability and, therefore, emptiness problems are $\WE\+\LEN$ and $\WE\+\REG\+\LEN$. Emptiness for the former in particular is equivalent to the satisfiability problem for word equations with length constraints which is a long-standing and important open problem in the field. Similarly, the latter is prominent in the context of string-solving and as such satisfiability/emptiness also presents an important open problem which is likely to be closely related to that of $\WE \+ \LEN$. 
Consequently, $\WE \+ \VPL$ presents a particularly interesting case as a ``reasonable'' generalisation of $\WE\+\REG\+\LEN$ and, in the absence of answers regarding this theory, it makes sense to consider the same problems for theories  with slightly more or slightly less expressive power. If we extend the expressive power as far as $\WE+\CF$, then undecidability is inherited directly from $\CF$. However, satisfiability and emptiness remain decidable for $\VPL$. Moreover, visibly pushdown languages share many of the desirable computational properties of regular languages, meaning that we can view  $\WE \+ \VPL$ as a slighter generalisation of $\WE\+\REG\+\LEN$. Nevertheless, due to our result from Theorem \ref{thm:VPLCONequalsRE} of the previous section, we know that $\VPL \+ \CON$ expresses already $\RE$, so
emptiness and finiteness are undecidable for $\VPL \+ \CON$, and consequently for $\WE \+ \VPL$ and other families $\mathfrak{F}$ of theories satisfying $\VPL \+ \CON \preceq \mathfrak{F}$. The left part of Figure~\ref{fig:overviewoftheories} summarizes the understanding of the emptiness and finiteness problems, as resulting from our results. \looseness=-1

\section{Universality, Greibach's Theorem, and Expressibility Problems}\label{sec:universality}

Universality is an important problem for a number of reasons. Firstly, undecidability of universality implies undecidability of equivalence and inclusion for any theory in which the universal language $\Sigma^*$ is expressible (which is true in any string-based theory containing at least one tautology). Secondly, an undecidable universality problem is the foundation for Greibach's theorem, which is helpful for proving that many other problems are undecidable. E.g., we shall make use of Greibach's theorem to show several problems concerning expressibility of languages in different theories are undecidable. We recall Greibach's theorem below. \looseness=-1

\begin{theorem}[\cite{HopcroftUllman}]\label{thm:greibachOriginal}
Let $\mathcal{C}$ be a class of formal languages over an alphabet $\Sigma \cup \{\#\}$ such that each language in $\mathcal{C}$ has some associated finite description. Suppose $\mathcal{P} \subsetneq \mathcal{C}$ with $\mathcal{P} \not= \emptyset$ and suppose that all the following hold:
\begin{enumerate}
\item{}
$\mathcal{C}$ and $\mathcal{P}$ both contain all regular languages over $\Sigma \cup \{\#\}$,
\item{}
$\mathcal{P}$ is closed under quotient by a single letter,
\item{}
Given (descriptions of) $L_1,L_2 \in \mathcal{C}$ descriptions of $L_1  \cup  L_2$, $L_1R$ and $RL_1$ can be computed for any regular language $R \in \mathcal{C}$,
\item{}
It is undecidable whether, given $L \in \mathcal{C}$, $L = \Sigma^*$.
\end{enumerate}
Then the problem of determining, for a language $L \in \mathcal{C}$, whether $L \in \mathcal{P}$ is undecidable.
\end{theorem}

Note that in order to apply Greibach's theorem, we need a variant of the universality problem to be undecidable which refers to a sub-alphabet, rather than the whole alphabet. 

\begin{definition}
Let $\mathcal{T}$ be a theory defined in Section~\ref{sec:theories} with underlying alphabet $\Sigma$ and such that $|\Sigma| \geq 3$. The subset-universality problem is: given a formula $f \in \mathcal{T}$, variable $x$ occurring in $f$ and $S \subset \Sigma$ with $|S| > 1$, is the language expressed by $x$ in $f$ equal to $S^*$?
\end{definition}

We can infer the following results already from the literature:

\begin{theorem}[\cite{Ganesh2012WordEW,durnev1995, Alur2004, HopcroftUllman}]\label{the:universalityliterature}
Universality is undecidable for $\WE$ and $\CF$, and decidable for $\VPL$. Subset-universality is decidable for $\VPL$ but not $\CF$.
\end{theorem}

To discuss the equivalence and inclusion problems, it makes sense to consider them in a general setting where the two languages may be taken from different theories. We therefore consider equivalence and inclusion problems for pairs of theories $(\mathfrak{T}_1, \mathfrak{T}_2)$. Combining the known results above with the constructive equivalences pointed out in Remark~\ref{rem:constructiveequivalence}, we easily get that
equivalence and inclusion for $(\mathfrak{T}_1, \mathfrak{T}_2)$ are undecidable whenever at least one of $\mathfrak{T}_1,\mathfrak{T}_2$ contains $\WE$ or $\CF$, but they are decidable for all other pairs of theories. In a similar way, one can show that cofiniteness is undecidable for $\WE$.

We may, clearly, propagate undecidability of universality and related problems upwards through families of theories containing $\WE$ (or $\CF$) as a syntactic subset, or apply Rice's theorem to get such results for all theories expressing $\RE$.
We can also show the following.\looseness=-1

\begin{theorem}\label{the:subsetuniversality}
Subset-universality is decidable for $\WE \+ \LEN$ and undecidable for $\WE + \REG$. In particular, for $S$ large enough, for any theory $\mathfrak{T}$ from $\WE \+ \REG$ with underlying alphabet $\Sigma \supset S$, the problem of whether a language expressed in $\mathfrak{T}$ is exactly $S^*$ is undecidable.\looseness=-1
\end{theorem}
\begin{proof}
The undecidability result can be obtained as follows. Recall from Theorem~\ref{the:universalityliterature} that the standard universality problem for $\WE$ is undecidable. In particular, there is an alphabet $S$ and theory $\mathfrak{T}'$ from $\WE$ whose underlying alphabet is $S$ such that the universality problem for $\mathfrak{T}'$ is undecidable. Let $\Sigma \supset S$. We reduce the universality problem for $\mathfrak{T}'$ to the subset universality problem for the the theory $\mathfrak{T}$ from $\WE \+ \REG$ whose underlying alphabet is $\Sigma$. In particular, for a formula $f'$ from $\mathfrak{T}'$ with variables $x_1,x_2,\ldots,x_n$, we construct the formula $f = f' \land \bigwedge\limits_{1 \leq i \leq n} x_i \in S^*$. Clearly, $f$ and $f'$ have exactly the same set of satisfying assignments, and so the language expressed by a variable $x_i$ is unchanged. Consequently, the language expressed by $x_i$ in $f'$ (in $\mathfrak{T}'$) is universal if and only if the language expressed by $x_i$ in $f$ satisfies the subset-universality problem for the subset $S$ of $\Sigma$. Hence the subset universality problem is undecidable for $\WE \+ \REG$.

Next we consider the same problem for $\WE$ and $\WE \+ \LEN$. Note firstly that, for any theory $\mathfrak{T}$ from $\WE \+ \LEN$ whose underlying alphabet is $\Sigma$, if $1 < |S| < |\Sigma|$, then it follows from the same construction as in the proof of Lemma~\ref{lem:subalphabetexample} that $S^*$ is not expressible by any formula/variable from $\mathfrak{T}$. Thus in any such case we can simply automatically answer ``no''. 
\end{proof}

Theorem~\ref{the:subsetuniversality} allows us to apply Greibach's Theorem to many theories defined in Section~\ref{sec:theories}.

\begin{theorem}\label{the:greibachapplies}
Let $\mathfrak{F}$ be a family of theories defined in Section~\ref{sec:theories} containing $\WE \+ \REG$. For large enough alphabets $\Sigma$, if $\mathcal{C}$ is the class of languages expressible by the theory $\mathfrak{T} \in \mathfrak{F}$ with underlying alphabet $\Sigma$, then the conditions of Greibach's theorem are satisfied by $\mathcal{C}$.
\end{theorem}

\begin{proof}
Recall that we consider descriptions of languages expressible in a theory to be a formula $f$ together with a variable $x$ occurring in $f$. If $f_1,x_1$ and $f_2,x_2$ describe languages $L_1$ and $L_2$ respectively, then $L_1 \cup L_2$ is described by the formula $f_1 \lor f_2'$ where $f_2'$ is obtained by renaming the variables in $f_2$ such that $x_2$ is renamed to $x_1$ and all other variables in $f_2$ do not occur in $f_1$. 

Now, suppose we have a family of theories $\mathfrak{F}$ which includes $\WE \+ \REG$. Let $\mathfrak{T} \in \mathfrak{F}$ have underlying alphabet $\Sigma \cup \{\#\}$ for ``large'' $\Sigma$. Then every regular language $R$ over $\Sigma \cup \{\#\}$ is expressible via the variable $x$ in the formula $x \in R$. Moreover for any language $L$ expressible by a formula $f$ and variable $x$, we may express the languages $LR$ and $RL$ through the variable $y$ in the formulas $f \land y = x z \land z \in R$ and $y = f \land z x \land z \in R$ respectively where $y,z$ are new variables not in $f$.

Finally, we note that by Theorem~\ref{the:subsetuniversality}, the problem of whether a given language expressible in $\mathfrak{T}$ is exactly $\Sigma^*$ is undecidable. Thus, we have shown all the conditions of Greibach's theorem applicable to $\mathcal{C}$ hold when $\mathcal{C}$ is the class of languages expressible in $\mathfrak{T}$.
\end{proof}

In the following, we give an example application of Theorem~\ref{the:greibachapplies} with respect to the pumping lemma for regular languages (see e.g.~\cite{HopcroftUllman}). Aside from defining an interesting superclass of the regular languages itself, there are many reasons to be interested in notions of pumping. For example, when considering (in)expressibility questions (even beyond the regular languages), as well as part of a strategy for producing satisfiability results in the context of length constraints or other restrictions. We use the pumping lemma for regular languages because it is well known, but the ideas are easily adapted to other useful notions of pumping and closure properties more generally. We recall first this lemma. 
\begin{lemma}[\cite{HopcroftUllman}]
Let $L$ be a regular language. Then there exists a constant $c$ such that for every $w \in L$ with $|w| > c$, there exist $x,y,z$ such that
(i) $|xy| < c$, and
(ii) $w = xyz$, and
(iii) $xy^nz \in L$ for all $n \in \mathbb{N}_0$.
\end{lemma}

Now, we can show how Theorem~\ref{the:greibachapplies} can be applied in this context.
\begin{theorem}\label{the:greibachpumping}
It is undecidable whether a language expressed by a formula in a theory from $\WE + \REG$ satisfies the pumping lemma for regular 
languages.
\end{theorem}

\begin{proof}
By Theorem~\ref{the:greibachapplies}, there exist theories $\mathfrak{T}$ from $\WE\+ \REG \+ \LEN$ such that the class $\mathcal{C}$ of languages expressible in $\mathfrak{T}$ satisfy the conditions of Greibach's theorem. It remains to observe that the property $\mathcal{P}$ of being a language which is expressible in $\mathfrak{T}$ which also satisfies the pumping lemma for regular languages also satisfies the conditions of Greibach's theorem. In particular, we note that all regular languages over the underlying alphabet belong to $\mathcal{P}$ since all regular languages satisfy the pumping lemma for regular languages. Moreover, it is straightforward to show that given a letter $a$, and a language $L$ satisfying the pumping lemma, that the quotient of $L$ by $a$ also satisfies the pumping lemma. Since we can express the same quotient in $\WE \+ \REG$ with the formula $f \land y a = x$ where $f$ and $x$ are the formula and variable expressing $L$ respectively, we conclude that $\mathcal{P}$ is closed under quotient by a single letter. Thus for an appropriate choice of underlying alphabet for $\mathfrak{T}$, we may apply Greibach's theorem to get the claimed result.
\end{proof}

Theorem~\ref{the:subsetuniversality} also tells us that we cannot use Greibach's theorem as stated to show that properties of languages expressible in $\WE \+ \LEN$ are undecidable. We leave as an open problem whether an equivalent of Greibach's theorem can be adapted to this context:

\begin{openproblem}
Is there an equivalent of Theorem \ref{thm:greibachOriginal} for the classes of languages expressible in $\WE \+ \LEN$ or $\WE$?\looseness=-1
\end{openproblem}

\section{Expressivity Problems}
\label{sec:expressivity}

Further, we consider decision problems related to expressivity. These problems have the general form: given a language $L$ expressed by a formula in a theory $\mathfrak{T}_1$ and given a second theory $\mathfrak{T}_2$, can we decide whether or not $L$ can be expressed by a formula in $\mathfrak{T}_2$? 

We begin by noting that since it is decidable whether or not a deterministic context-free language is regular (see~\cite{stearns1967,valiant1975}), the same holds true for visibly pushdown languages, and hence whether a language expressed in $\VPL$ can be expressed in $\REG$.
Therefore, it is clearly decidable whether a language expressed in $\VPL$ is expressible in $\REG$. The same holds for theories from families equivalent to $\VPL$ and $\REG$ under the relation~$\sim$.

Naturally, since we have already seen that $\VPL \+ \CON$ is capable of expressing all $\RE$-languages, it is undecidable whether a language expressed in a theory from $\VPL \+ \CON$ is expressible in a theory from any of the families which have strictly less expressive power.

The separation results from Section~\ref{sec:separation} and Theorem~\ref{the:greibachapplies} together mean we can get the following negative results as a consequence of Greibach's theorem. They have a particularly relevant interpretation in the context of string solving in practice. Specifically, it is often the case that string-solvers will perform some pre-processing of string constraints in order to put them in some sort of normal form which will make them easier to solve. One natural thing to want to do in this process is to reduce the number of combinations of sub-constraints of differing types by converting constraints of one type to another. This is useful particularly in cases where the combinations are difficult to deal with together in general. Word equations, regular constraints and length constraints are one such combination (recall from Section~\ref{sec:separation} that satisfiability for the corresponding theory including all three types of constraint is an open problem, but if length constraints are removed then satisfiability becomes decidable). Unfortunately, the following theorem reveals that we cannot in general decide whether length constraints can be eliminated by rewriting them using only regular membership constraints and word equations.\looseness=-1

\begin{theorem}\label{the:eliminatinglength}
It is undecidable whether a language expressed in $\WE+\REG + \LEN$ can be expressed in $\WE+\REG$.
\end{theorem}

\begin{proof}
Theorem~\ref{the:greibachapplies} tells us that there exist theories $\mathfrak{T}$ from $\WE \+ \REG \+ \LEN$ such that the class $\mathcal{C}$ of languages expressible in $\mathfrak{T}$ satisfy the conditions of Greibach's theorem. It remains to observe that there exist theories $\mathfrak{T}'$ from $\WE \+ \REG$ such that the property  $\mathcal{P}$ of being a language expressible in $\mathfrak{T}'$ also satisfies the conditions of Greibach's theorem. 

In particular, we note that all regular languages over the underlying alphabet belong to $\mathcal{P}$ trivially through the use of regular language membership constraints. Moreover, it is straightforward that given a letter $a$, and a language $L$ expressible in $\mathfrak{T}'$, we can express the quotient of $L$ by $a$ in $\mathfrak{T}'$ with the formula $f \land y a = x$ where $f$ and $x$ are the formula and variable expressing $L$ respectively and where $y$ is a new variable. Thus $\mathcal{P}$ is closed under quotient by a single letter. For an appropriate choice of $\mathfrak{T}'$ (i.e. for the appropriate underlying alphabet), we may hence apply Greibach's theorem to get the claimed result.
\end{proof}





The same undecidability result holds if, instead of removing length constraints by rewriting them as regular membership constraints and word equations, we want to remove word equations constraints by rewriting them as regular language membership constraints (possibly also with length constraints which, in the absence of word equations, do not increase the expressive power due to Theorem~\ref{the:regulartheories}). While this result can also be obtained via Greibach's theorem, we can, in fact, state a stronger version for which we need a novel approach. In particular, we show that it is already undecidable whether a language expressible by word equations (without additional constraints) is a regular language (i.e., can be expressed in $\REG$).\looseness=-1

\begin{theorem}\label{the:undecifWEisREG}
It is undecidable whether a language expressed in $\WE$ is regular. 
In other words, it is undecidable whether a language expressed by a formula from $\WE$ is regular. 
\end{theorem}
\begin{proof}
We shall prove the statement by giving a reduction from the problem of determining whether or not the set of words belonging to $0^+$ accepted by a 2-Counter Machine (2CM) is finite. Since we shall use word equations to model computations of 2CMs, our proof has a similar flavour to the one in~\cite{Ganesh2012WordEW}, but since our aims and setting are different, the details and our construction is also necessarily different.

A 2CM $M$ is a deterministic finite state machine with 3 semi-infinite storage tapes, each with a leftmost cell but no rightmost cell. One is the input tape, on which the input is initially placed. There is a read-only head which can move along the input tape in both directions but cannot move beyond the input word and cannot overwrite it. The other two tapes represent counters. They each store a non-negative integer represented by the position of a head which can move to the left or right. If the head is in the leftmost position, the number represented is 0, and increments of one are achieved by moving the head one position to the right.

$M$ can test if each counter is empty but cannot compare directly the stored numbers for equality. It accepts a word if the computation with that word as input terminates in an accepting state and such that all tape heads (input and both counters) are at the leftmost position. Formally, a 2CM is a tuple $(Q, \Delta, \delta, q_0, F)$ where:
\begin{enumerate}
    \item $Q$ is a finite set of states, $q_0 \in Q$ is an initial state and $F \subseteq Q$ is a set of final or accepting states.
    \item $\Delta$ is a finite input tape alphabet.
    \item $\delta : Q \times \Delta \times \{T,F\} \times \{T,F\} \to Q \times \{1,2,3\} \times \{L,R\}$ is a transition function.
\end{enumerate}
The interpretation of the transition function is as follows: $\delta(q,a,Z_1,Z_2) = (q',i,D)$ if before the transition $M$ is in state $q$ and currently reads letter $a$ on the input tape, and $Z_1$ and $Z_2$ are $T$ if the first and second counters are $0$ respectively and $F$ otherwise, and after the transition $M$ is in state $q'$, $D$ indicates the direction in which one of the tape heads moves (L for left and R for right), and $i$ determines which tape head moves (1 for input head and 2 and 3 for the first and second counters respectively).

We can represent a configuration of a 2CM at any point in a computation as a word belonging to $Q \Delta^* a^+ b^+ c^+$ (assuming $a,b,c$ are new letters such that $Q, \Delta, \{a,b,c\}$ are pairwise disjoint) where the leftmost letter (from $Q$ is the current state, the part from $\Delta^*$ stores the contents of the input tape (so, the input), and the a's b's and c's denote in unary notation the position of the input tape head and the values of the two counters. For convenience, we add one to all these values so that the sequences of a's, b's and c's are all non-empty.

An initial configuration on input $w\in \Delta^*$ has the form $q_0wabc$, and a final configuration has the form $q_fwabc$ where $q_f \in F$.

A valid computation history of a 2CM $M$ on input word $w$ is a finite word $C = C_1C_2C_3\ldots C_n$ such that each $C_i$ is a configuration, $C_1$ is the initial configuration for the input $w$, $C_n$ is a final configuration, and such that each successive pair of configurations $C_i,C_{i+1}$ respects the transition function $\delta$ of $M$.

Now we can explain the first main step of our proof.

It is well known that 2CMs can simulate the computations of Turing Machines, and therefore that they accept the class of recursively enumerable languages. Hence, we obtain from Rice's theorem that it is undecidable whether the language accepted by a 2CM contains infinitely many words from $\{0\}^+$ or not, where $0 \in \Delta$. Moreover, since 2CMs are deterministic, each word accepted by a given 2CM has exactly one valid computation history. So, it follows that the set of words from $\{0\}^+$ accepted by a 2CM $M$ is finite if and only if the set $S_M = \{ C \mid C $ is a valid computation history for  $M$  on  some input word  $w \in \{0\}^+\}$ is finite.

Moreover, it is easily seen that the set $S_M$ is finite if and only if it is regular. Clearly, if it is finite, it is regular. To see the converse, suppose for contradiction that it is both infinite and regular. Then there is a DFA accepting $S_M$, and moreover, there exist arbitrarily long words in $S_M$. If we choose a word which is sufficiently long, then there must exist distinct positions in that word which occur after the initial configuration and such that the DFA accepting $S_M$ must be in the same state after reading the prefix up to those positions. It follows that the part between these two positions can be pumped any number of times without affecting acceptance (and hence membership in $S_M$). However, this would either disrupt the correct form of a valid computation history for the particular input word (which is fixed by the first configuration), or it would result in infinitely many valid computation histories for the same input word, which contradicts the fact that $M$ is deterministic. In both cases, we get a contradiction, so if $S_M$ is infinite, it is not regular.

We can now proceed with the second main part of our proof, which is essentially a computation-simulation. Interestingly, our tools for this part of the proof are fundamentally different from those used in Theorem \ref{thm:VPLCONequalsRE}, as we now have to rely on techniques rooted in combinatorics on words (and word equations) rather than on techniques related to (visibly pushdown) automata. 

Next, we note that $S_M$ is regular if and only if its complement is regular. In what remains, we shall construct, for any given 2CM $M$, a $\WE$ formula $f$ containing a variable $x$ such that the language expressed by $x$ in $f$ is exactly the complement of $S_M$. This construction facilitates a reduction from the finiteness problem described at the beginning of the proof to the problem of whether or not the language expressed by a variable in a $\WE$ formula is regular.

Let us fix a 2CM $M$. We construct the formula $f$ as the disjunction of 4 subformulas, each of which accounts for a particular way in which a word substituted for $x$ could violate the definition of a valid computation history of $M$ on an input of the form $0^+$. Let $x,y_1,y_2,y_3,y_4,z_1,z_2,z_3,z_4,u,u',v,v',w,w'$ be variables.

Throughout the construction we shall repeatedly use the well-known fact (see the Defect Theorem in~\cite{lothaire1997} that for two words $w_1,w_2$, we have $w_1w_2 = w_2 w_1$ if and only if they are repetitions of the same word, that is there exists a word $w_3$ and $p,q \in \mathbb{N}$ such that $w_1 = w_3^p$ and $w_2 = w_3^q$.

Now, $f$ is the formula
\[w 0 = 0 w \land w= 0 w' \land (f_1 \lor f_2 \lor f_3 \lor f_4)\]
where $f_1,f_2,f_3,f_4$ are defined below. The first two conjuncts enforce that $w \in 0^+$. This allows us to use  $w$ to represent the input word in the rest of the formula. $f_1$ will be satisfiable for a given value of $x$ if $x$ does not belong to $q_0 0^+a^+b^+c^+(Q0^+a^+b^+c^+)^*$, and thus that it is not a sequence of configurations of $M$ starting in an initial state. $f_2$ and $f_3$ will cover the cases when $x$ does not start with the initial configuration for $M,w$ and when $x$ does not end with a final configuration respectively. Finally $f_4$ will cover the case that two consecutive configurations in $x$ do not respect the transition relation $\delta$.

Let $P$ be the set of pairs of letters which may not occur consecutively in $(Q0^+a^+b^+c^+)^+$. That is, $P$ is the complement of $Q0 \cup cQ \cup \{00, 0a, aa, ab, bb, bc, cc\}$. Note that $x$ is not in the language $q_0 0^+a^+b^+c^+(Q0^+a^+b^+c^+)^*$ if and only if it contains consecutive letters included in $P$ or it starts with a letter other than $q_0$ or it ends with a letter other than $c$. Note that $P$ is finite. Thus the subformula $f_1$ is given by
\[ \bigvee\limits_{AB \in P} x = u AB v \lor \bigvee\limits_{A \in \Sigma \backslash \{q_0\}} x = A u \lor \bigvee\limits_{A \in \Sigma \backslash F} x = u A \lor x = \varepsilon.\]

With $f_2$, we want to enforce that it is true only if $x$ has a prefix other than the initial configuration, namely $q_0 w abc$. We only need to cover cases when $f_1$ is not satisfied (so we may assume that $x$ belongs to $q_0 0^+a^+b^+c^+(Q0^+a^+b^+c^+)^*$. Thus $f_2$ is given as:

\begin{align*} & u0 = 0u \\
&\land ( \bigvee\limits_{A_1A_2A_3 \not= abc \land A_1 \not= 0} x = q_0 u A_1 A_2 A_3 v  \\ 
&\lor (x = q_0 u av \land (w = u 0 u' \lor w0u' = u))).
\end{align*}

In the above formula, $u$ must be the complete sequence of $0$s occurring after $q_0$. The cases when the next three letters after $u$ are not $abc$ and when $u\not= w$ are then covered by the second and third lines.

$f_2$ can be constructed similarly as follows:
\[ u0 = 0u \land \left( \bigvee\limits_{A_0A_1A_2A_3 \not= 0abc} x = v A_0 A_1 A_2 A_3 \lor \bigvee\limits_{q \in Q\backslash F} x = v q u abc \right).\]
The first of the two disjuncts inside the brackets covers all cases when $x$ does not end with a configuration of the form $q0^*abc$, or in other words when all tape heads have not returned to their leftmost positions. The second disjunct covers the cases when tape heads are in their leftmost positions but the state is not final.

Finally we construct $f_4$ as 
\begin{align*}
    \bigvee\limits_{q,q' \in Q} &( x = u q y_1 y_2 y_3 y_4 q' z_1 z_2 z_3 z_4 v \land \\
    & y_1 0 = 0 y_1 \land \\
    & y_2 a = a y_2 \land \\
    & y_3 b = b y_3 \land \\
    & y_4 c = c y_4 \land \\
    & z_1 0 = 0 y_1 \land \\
    & z_2 a = a z_2 \land \\
    & z_3 b = b z_3 \land \\
    & z_4 c = c z_4 \land \\
    & (y_1 0 u' = z_1 \lor y_1 = z_1 0 u' \lor \\
    & y_2 aa u' = z_2 \lor y_2 = z_2 aa u' \lor \\
    & y_3 bb u' = z_3 \lor y_3 = z_3 bb u' \lor \\
    & y_4 cc u' = z_4 \lor y_4 = z_4 cc u' ) \lor \\
    &  \bigvee\limits_{f' \in D} f')
\end{align*}

where $D$ is a set of formulas describing transitions which are not possible in $M$, which is again given below. Essentially the first 9 lines of $f_4$ enforces that $q y_1 y_2 y_3 y_4$ and $q' z_1 z_2 z_3 z_4$ are consecutive configurations in $x$ and that $q,q'$ represent states, while $y_1,z_1$ are the part containing $0$'s, $y_2,z_2$ contain the $a$'s $y_3,z_3$ contain the $b$'s and $y_4z_4$ contain the $c$'s. The 10th line accounts for when the input word is not correctly copied from the first configuration to the next, while the 11th, 12th and 13th lines account for when one of the tape heads moves two or more positions.

Thus the formulas from $D$ must cover the cases when the input is copied correctly and all tape heads move at most one position, but the transition is still not valid. Thus $D$ contains the following formula where $B_2 = a, B_3 = b, B_4 = c$, which covers the case that two tape heads move at the same time.

\begin{align*}
\bigvee\limits_{i,j \in \{2,3,4\} \land i \not=j}
    ((y_iB_i = z_i \lor y_i = z_i B_i) \land (y_j B_j = z_j \lor y_j = z_j B_j))
\end{align*}

Moreover, for every "legal" transition not specified by $\delta$ (so every transition having the correct form but not allowed in the specific 2CM $M$), $D$ contains a formula describing this transition. We provide an example in the case that the transition $\delta(q_1,0,T,F) = (q_2, 2, L)$ is present in $M$. The example can easily be adapted for other combinations. The formula will enforce that if the current state in the first configuration is $q_1$, and the first counter is $0$ while the second is not (we only need to consider cases where the input letter is $0$), then the next configuration should not have state $q_2$, or should increment the second counter or leave it unchanged, or should change either the first counter or the input tape head position. This covers all possible ways the transition $\delta(q_1,0,T,F) = (q_2, 2, L)$ is not respected. Formally the formula is given as:

\begin{align*}
    &q = q_1 \land y_3 = b \land y_4 = bb v' \land \\
    & (\bigvee\limits_{q_2' \in Q \backslash \{q_2\}} q' = q_2' \lor \\
    & y_2 = z_2 a \lor y_2 a = z_2
    \lor\\
    & y_4 = z_4 c \lor y_4 c = z_4 \lor \\
    & y_3 = z_3 \lor y_3 = z_3 b).
\end{align*}
The first line establishes the appropriate conditions of the first configuration, while the second line covers cases where the second configuration has the wrong state, the third and fourth lines covers the case that one of the wrong tape heads moves while the final line covers the case that the correct tape head moves in the wrong direction or not at all. Inclusion of similar formulas in $D$ for other transitions completes the construction of $f$.

All together, we have shown a construction for a formula $f$ which can be satisfied for a particular value of $x$ if and only if $x$ is not a valid computation history for $M$ on input word of the form $0^+$. In other words, the language expressed by $x$ in $f$ is exactly the complement of $S_M$. Thus it is regular if and only if $M$ accepts only finitely many words from $0^+$. This completes the reduction and we may conclude that the problem of deciding whether a formula and variable from $\WE$ express a regular language is undecidable as claimed.
\end{proof}

Although a trivial consequence of Theorem~\ref{the:undecifWEisREG}, is is somehow surprising that it remains undecidable if a word equation combined with regular constraints expresses a regular language.

Finally, we note the remaining cases which correspond to removing regular language membership constraints in the presence of word equations, and removing length constraints in the presence of word equations but without regular constraints.  Thus, we leave the following questions open:

\begin{openproblem}
Is it decidable whether a language expressed in $\WE+\REG$ (respectively, $\WE + \REG + \LEN$) can be expressed in $\WE$ (respectively, $\WE + \LEN$)? Is it decidable whether a language expressed in $\WE + \LEN$ can be expressed in $\WE$?
\end{openproblem}

\bibliography{New_version}

\end{document}